\documentclass[journal]{IEEEtran}

\usepackage{array}
\usepackage{booktabs}
\usepackage{graphicx}
\usepackage{amsmath}
\usepackage{amsfonts}
\usepackage{amssymb}
\usepackage{pifont}
\usepackage{amsthm}
\usepackage{multirow}
\usepackage{anyfontsize}
\usepackage{cite}
\usepackage[table]{xcolor}
\usepackage[linesnumbered,ruled,vlined]{algorithm2e}

\usepackage[font = footnotesize]{caption}
\usepackage{subcaption}

\usepackage{blindtext}
\newcommand\blfootnote[1]{%
\begingroup
\renewcommand\thefootnote{}\footnote{#1}%
\addtocounter{footnote}{-1}%
\endgroup
}

\DeclareMathAlphabet\mathbfcal{OMS}{cmsy}{b}{n}
\DeclareMathOperator{\Trace}{Tr}

\newtheorem{lemma}{Lemma}
\newtheorem{corollary}{Corollary}

\DeclareMathOperator*{\argmin}{arg\,min}
\DeclareMathOperator*{\diag}{diag}

\usepackage{amsfonts}

\usepackage{makecell}

\begin{document}

\title{Plug-In RIS: A Novel Approach to Fully Passive Reconfigurable Intelligent Surfaces}

\author{Mahmoud~Raeisi,~\IEEEmembership{Student Member,~IEEE},
        Ibrahim~Yildirim,~\IEEEmembership{Graduate Student Member,~IEEE},
        Mehmet C.~Ilter,~\IEEEmembership{Senior Member,~IEEE},
        Majid~Gerami,~\IEEEmembership{Member,~IEEE},
        and
        Ertugrul~Basar,~\IEEEmembership{Fellow,~IEEE}
\thanks{M. Raeisi, I. Yildirim, and E. Basar are with the Communications Research and Innovation Laboratory (CoreLab), Department of Electrical and Electronics Engineering, Koç University, Sariyer, Istanbul 34450, Turkey. (e-mail: mraeisi19@ku.edu.tr; yildirimib@itu.edu.tr; ebasar@ku.edu.tr)}

\thanks{I. Yildirim is also with the Faculty of Electrical and Electronics Engineering, Istanbul Technical University, Istanbul 34469, Turkey.}

\thanks{Mehmet C. Ilter was with Huawei Lund Research Center during this work and is currently with the Department of Electrical Engineering, Tampere University, Finland (mehmet.ilter@tuni.fi).}

\thanks{M. Gerami is with the Lund Research Center, Huawei Technologies Sweden AB, Sweden. (e-mail: majid.gerami@huawei.com)}
\vspace{-3em}
}

\markboth{IEEE ***,~Vol.~***, No.~***, MONTH~YEAR}%
{Raeisi \MakeLowercase{\textit{et al.}}: plug-in Reconfigurable Intelligent Surface}

\maketitle
\vspace{-2em}

\begin{abstract}
This paper presents a promising design concept for reconfigurable intelligent surfaces (RISs), named plug-in RIS, wherein the RIS is plugged into an appropriate position in the environment, adjusted once according to the location of both base station and blocked region, and operates with fixed beams to enhance the system performance.
The plug-in RIS is a novel system design, streamlining RIS-assisted millimeter-wave (mmWave) communication without requiring decoupling two parts of the end-to-end channel, traditional control signal transmission, and online RIS configuration. In plug-in RIS-aided transmission, the transmitter efficiently activates specific regions of the divided large RIS by employing hybrid beamforming techniques, each with predetermined phase adjustments tailored to reflect signals to desired user locations. This user-centric approach enhances connectivity and overall user experience by dynamically illuminating the targeted user based on location.
By introducing plug-in RIS's theoretical framework, design principles, and performance evaluation, we demonstrate its potential to revolutionize mmWave communications for the limited channel state information (CSI) scenarios. Simulation results illustrate that plug-in RIS provides power/cost-efficient solutions to overcome blockage in the mmWave communication system and a striking convergence in average bit error rate and achievable rate performance with traditional full CSI-enabled RIS solutions.
\end{abstract}

\vspace{-1em}


\vspace{-1em}

\IEEEpeerreviewmaketitle

\vspace{-0.5em}
\section{Introduction}
\vspace{-0.5em}

\IEEEPARstart{T}{he} requirements posed by the next generation of communication networks, including massive connectivity, ultra-reliability, low latency, and energy efficiency, necessitate exploring innovative strategies and technologies to enhance existing communication systems, enabling the deployment of ubiquitous connected devices. Reconfigurable intelligent surfaces (RISs), a promising technology that acquired substantial attention from research community, has emerged as a potential candidate for incorporation into the next-generation wireless communication networks \cite{basar2019wireless}. By intelligently manipulating incident signals, RISs enable a smart environment through controlled scattering paths, significantly improving performance. The RIS-aided virtual channel between the terminals, making the channel adaptable for boosting communication performance in various scenarios or circumventing blockages that impede communication, a common occurrence in millimeter-wave (mmWave) communication systems due to signal vulnerability to the blockage and severe path loss. These blocked regions, termed dead zones, can be addressed by properly configuring the RIS phase shifts based on channel state information (CSI), allowing the RIS to redirect signals around obstructions and serve users within the dead zones.
It is worth mentioning that CSI acquisition and phase shift configuration are key limiting factors significantly affecting the RIS design and system performance.
Hence, the following sub-section elaborates on different RIS designs in the existing literature, mostly focusing on CSI acquisition methods for RIS phase adjustment \cite{8683663, 9103231, 10053657, 9322519, 8879620, 9400843, 9149146, 9614196, 9681847, 9685434, 9370097, 9529045, 10065831}.

\begin{table*}
\centering
    \captionsetup{justification=centering}
    \caption{Comparison among different RIS structures available in the existing literature and the proposed plug-in RIS. 
(\textbf{Abbreviations}: Y = yes; N = no; FP = fully passive; SP = semi-passive; H = high; M = medium; L = low.)}

\arrayrulecolor{cyan}
\begin{tabular}{ | m{3.5 cm} || c c c c c c c c c c c c c | m{1cm} | }
   \hline
 \rowcolor{cyan!20} & \textbf{\cite{8683663}} & \textbf{\cite{9103231}} & \textbf{\cite{10053657}} & \textbf{\cite{9322519}} & \textbf{\cite{8879620}} & \textbf{\cite{9400843}} & \textbf{\cite{9149146}} & \textbf{\cite{9614196}} & \textbf{\cite{9681847}} & \textbf{\cite{9685434}} & \textbf{\cite{9370097}} & \textbf{\cite{9529045}} & \textbf{\cite{10065831}} & \cellcolor{teal!20}\textbf{Plug-in RIS} \\ 
 \hline\hline

 
 \cellcolor{cyan!20} \textbf{RIS type}: & FP & 
 FP & 
 FP &
 FP &
 FP &
 FP &
 FP & 
 FP &
 FP &
 SP &
 SP & 
 SP & 
 SP &
 \cellcolor{teal!10} FP \\ 
 
 \cellcolor{cyan!20} \textbf{Channel estimation carries out by}: & BS & UE & BS & UE & UE & BS & UE & BS & BS & RIS & RIS & BS/UE & BS/UE &\cellcolor{teal!10} BS\\

 \hline \hline
 
 \cellcolor{cyan!20} \textbf{Control link availability}: & \cellcolor{cyan!5}\color{red}{Y} & \cellcolor{cyan!5}\color{red}{Y} & \cellcolor{cyan!5}\color{red}{Y} & \cellcolor{cyan!5}\color{red}{Y} & \cellcolor{cyan!5}\color{red}{Y} & \cellcolor{cyan!5}\color{red}{Y} & \cellcolor{cyan!5}\color{red}{Y} & \cellcolor{cyan!5}\color{red}{Y} & \cellcolor{cyan!5}\color{red}{Y} & \cellcolor{cyan!5}\color{teal}{N} & \cellcolor{cyan!5}\color{teal}{N} & \cellcolor{cyan!5}\color{red}{Y} & \cellcolor{cyan!5}\color{red}{Y} & \cellcolor{teal!10}\color{teal}{N}\\

 \cellcolor{cyan!20} \textbf{Implementation complexity}: & \cellcolor{cyan!5}\color{orange}{M} & \cellcolor{cyan!5}\color{orange}{M} & \cellcolor{cyan!5}\color{red}{H} & \cellcolor{cyan!5}\color{red}{H} &  \cellcolor{cyan!5}\color{red}{H} & \cellcolor{cyan!5}\color{orange}{M} & \cellcolor{cyan!5}\color{red}{H} & \cellcolor{cyan!5}\color{orange}{M} & \cellcolor{cyan!5}\color{orange}{M} & \cellcolor{cyan!5}\color{red}{H} & \cellcolor{cyan!5}\color{red}{H} & 
 \cellcolor{cyan!5}\color{orange}{M} & \cellcolor{cyan!5}\color{orange}{M} &
 \cellcolor{teal!10}\color{teal}{L} \\ 

 \cellcolor{cyan!20} \textbf{Signalling overhead}: & \cellcolor{cyan!5}\color{orange}{M} & \cellcolor{cyan!5}\color{orange}{M} & \cellcolor{cyan!5}\color{red}{H} & \cellcolor{cyan!5}\color{orange}{M} &  \cellcolor{cyan!5}\color{orange}{M} & \cellcolor{cyan!5}\color{orange}{M} & \cellcolor{cyan!5}\color{red}{H} & \cellcolor{cyan!5}\color{orange}{M} & \cellcolor{cyan!5}\color{orange}{M} & \cellcolor{cyan!5}\color{teal}{L} & \cellcolor{cyan!5}\color{teal}{L} & 
 \cellcolor{cyan!5}\color{orange}{M} & \cellcolor{cyan!5}\color{orange}{M} &
 \cellcolor{teal!10}\color{teal}{L} \\ 
 \hline
 
\end{tabular}
\label{Tab: Different RIS structures comparison}
\vspace{-1.8em}
\end{table*}

\vspace{-1.5em}

\subsection{Related Works}\label{Sec: Related works}
\vspace{-0.5em}
Two critical aspects distinguishing RIS deployment are its passive nature and ease of implementation, as highlighted in \cite{basar2019wireless}. Specifically, it is assumed that an ideal RIS does not require an external power source and can be easily adapted for implementation in various environments. 
By taking these attributes into account, we classify the existing RIS structures into two main categories: Fully passive RIS and semi-passive RIS\footnote{In the literature, there is a type of RIS called active RIS. However, active RIS is not equipped with baseband components like RF chains; it possesses the added capability of signal amplification \cite{9998527}. In respect of the channel estimation methodology, it is not different from the fully passive RIS.}.

\subsubsection{Fully passive RIS}

The RIS configuration depends on CSI availability for both the base station (BS)-RIS and RIS-user equipment (UE) channels, as assumed in most RIS studies in the literature. In a fully passive RIS setup, where the RIS lacks baseband processing capability, either the BS or the UE needs to estimate the end-to-end channel and decouple two parts of it, i.e., BS-RIS and RIS-UE.
Numerous studies in the existing literature have addressed the channel estimation problem for fully passive RIS systems \cite{8683663, 9103231, 10053657, 9322519, 8879620, 9400843, 9149146, 9614196, 9681847}.

In \cite{8683663}, the authors proposed a channel estimation method for passive RIS-assisted power transfer systems, assuming that the power beacon manages all computational tasks due to constraints in UE. {Besides, the RIS is deployed near the UE, and its controller is connected to the power beacon for programming and operating; therefore, a relatively long-distance dedicated control link is required to guarantee error-free control signaling.}
In \cite{9103231}, the authors studied channel estimation in an RIS-assisted mmWave communication system, aiming for joint active and passive beamforming for single-antenna UEs; nevertheless, dealing with multi-antenna UEs adds extra computational load for CSI acquisition and passive phase shift vector optimization.
The study of \cite{10053657} proposed a two-step channel estimation protocol for acquiring two parts of the end-to-end channel within a multi-user (MU) mmWave MIMO system, in which the BS is responsible for channel estimation. {In addition, the authors also considered time allocation and RIS phase shift adjustment for the uplink channel estimation phase.}
Meanwhile, the study of \cite{9322519} dives into channel estimation problem within a point-to-point (P2P) RIS-assisted mmWave communication system. Under the suggested protocol, the UE undertakes {a non-convex optimization problem to obtain BS-RIS and RIS-UE channels.}
Similarly, a channel estimation algorithm for passive large intelligent metasurfaces (LIM) (which is a very large RIS) is introduced in \cite{8879620}, employing sparse matrix factorization and matrix completion. {It is worth noting that running an algorithm in one of the transceivers is necessary to obtain BS-LIM and LIM-UE channels.}
In \cite{9400843}, the authors considered UE mobility and proposed a two-time scale channel estimation framework, leveraging the quasi-static BS-RIS channel while the RIS-UE is a low-dimensional mobile channel. {Accordingly, the BS-RIS channel with a large dimension needs to be estimated less frequently, while the small dimension RIS-UE channel must be estimated frequently. Consequently, the pilot overhead is reduced in a long-term perspective.}
The study of \cite{9149146} investigated channel estimation for broadband RIS-aided mmWave massive MIMO systems utilizing compressive sensing (CS) method. {Specifically, the authors assumed that the line-of-sight (LOS) dominated BS-RIS channel is known; hence, in order to jointly estimate BS-UE and RIS-UE channels, the pilot signals can be designed accordingly.} 
{The study of \cite{9614196, 9681847} proposed a novel framework for low-complex channel estimation and passive beamforming for both single-user and MU scenarios. The authors proposed that the composed superposed channel is estimated simultaneously instead of obtaining each direct link and reflective link separately. Additionally, a set of pre-adjusted training phase profiles can be adopted to acquire passive beamforming at the RIS.}

This literature review reveals that significant efforts have been made to estimate the end-to-end channel and separate it into BS-RIS and RIS-UE components in RIS-assisted mmWave communication systems. 
It is worth emphasizing that in fully passive RIS design, channel estimation/decoupling and optimizing phase shifts are done at one of the endpoints, which necessitates two key considerations.
Firstly, establishing a dedicated control link is essential to facilitate the transmission of configuration information from the responsible endpoint to the RIS. Secondly, allocating sufficient resources to accommodate the computational tasks associated with channel estimation/decoupling and RIS phase shift optimization at the responsible endpoint is essential.
However, it is worth noting that both of these requirements have potential drawbacks. 
Introducing a dedicated control link can add complexity to the system and potentially increase the overall latency stemming from reconfiguration latency \cite{etsigrris003}, while allocating substantial resources, particularly at the UE, may not be a cost/power-efficient solution. Similarly, if the BS assumes the responsibility for computational tasks, extending this approach to an MU scenario can significantly elevate system complexity and overhead, potentially posing challenges to the overall feasibility of the system.

\vspace{-0.3em}
\subsubsection{Semi-passive RIS}

Despite extensive endeavors to acquire individual cascaded channels for fully passive RIS configurations, challenges arise from the inefficient cascaded channel decoupling process due to the passive nature of RIS elements \cite{9685434}. In order to address this issue, semi-passive RIS structure is introduced in \cite{9370097}. 
In the semi-passive RIS design, a fraction of RIS elements are connected to baseband components, enabling the RIS to perform channel estimation through its integrated baseband components. 
Consequently, the RIS can accurately estimate BS-RIS and RIS-UE channels separately and adjust itself without being controlled by other terminals, i.e., the transmitter (Tx) or the receiver (Rx) \cite{9370097}. 
In \cite{9685434}, the authors proposed an algorithm for semi-passive RIS channel estimation, leveraging sparsity in both spatial and frequency domains of mmWave channels. As shown in \cite{9685434}, connecting only $8\%$ of RIS elements to the baseband components yields more accurate channel estimation than the considered benchmarks.
The study of \cite{9529045} proposed two practical residual neural networks to accurately estimate the channel for semi-passive RIS-aided communication systems. 
Similarly, the authors in \cite{10065831} proposed a channel estimation method for semi-passive RIS-assisted systems based on super-resolution neural networks. 
It is worth noting that even though the studies of \cite{9529045, 10065831} exploit semi-passive RIS, a significant portion of the computational load for channel estimation is carried out by the transceivers to simplify the RIS structure. Consequently, they inherit some of the drawbacks associated with both fully passive and semi-passive RIS designs.

Table \ref{Tab: Different RIS structures comparison} provides an overview of the structural characteristics of various RIS designs discussed above. The principal advantage of the semi-passive RIS design results from its reduced reliance on establishing a dedicated control link. Moreover, this approach eliminates the additional overhead and complexity imposed on the transceivers, as the RIS efficiently handles the computational tasks. 
Nonetheless, the semi-passive RIS configuration deviates from two fundamental characteristics highlighted previously: The passive nature of the RIS and its simplicity of implementation. While it is evident that the semi-passive RIS does not adhere to the passive attribute, it is crucial to remember that including baseband components introduces a costly and power-intensive structure, which, in turn, hinders easy and widespread deployment.

\vspace{-1em}

\subsection{Motivations and contributions}
\vspace{-0.5em}

As summarized in Table \ref{Tab: Different RIS structures comparison}, both fully passive and semi-passive RIS configurations exhibit drawbacks that pose challenges to their practical implementation.
Based on the existing literature, it becomes evident that there is a clear and pressing need within academia to create a novel RIS structure that fulfills the two fundamental attributes emphasized in \cite{basar2019wireless}, which strongly motivates us to introduce a fresh and innovative design, named \textit{plug-in RIS}, with the specific aim of simplified deployment and operation within mmWave communication environments. 
In this system, the RIS can be conveniently plugged into the environment to operate with fixed phase shift profiles, which are adjusted based on the location of the BS and dead zones.
Specifically, initially, the dead zone in the mmWave environment is defined to the BS (because the BS is initially not aware of which area is the dead zone) and divided into spatial segments; each spatial segment is considered to be served via a part of divided RIS, named sub-RIS. Next, we examine the location of each spatial segment and assign a fixed phase shift profile to the corresponding sub-RIS. Since both the BS and dead zone have fixed locations, the sub-RISs are not required to be configured frequently.
For activation, the BS obtains the UE's spatial segment and employs beamforming techniques to illuminate the corresponding sub-RIS, enabling control over signal reflection towards desired directions and offering a user-centric approach to enhance wireless communication performance.

Based on the aforementioned motivations, we outline our contributions as follows:
\vspace{-0.5em}
\begin{itemize}
    \item \textit{Plug-in RIS with standalone operation:} In contrast to the fully passive and semi-passive RIS designs, this study introduces an innovative approach—a plug-in RIS configuration with standalone operation—for extending communication coverage in mmWave systems. {Specifically, since mmWave communication is vulnerable to blockages, it can encounter several limited-space dead zones. In this context, developing an RIS design that facilitates widespread RIS implementation while keeping cost and complexity low and maintaining high-quality communication is essential for mmWave communication systems.} The proposed structure aligns with key attributes of RIS technology, such as, its passive nature and ease of implementation. The plug-in RIS consists of multiple sub-RIS planes, each pre-configured with fixed phase shifts, intended to cover distinct spatial segments within the dead zone. Leveraging the high beamforming gain at the {BS} in the mmWave communication systems, only one sub-RIS plane is activated in each transmission period to establish a virtual link toward the UE. This approach innovatively integrates the control mechanism within the transmitting beam, eliminating the necessity for implementing traditional control links and the associated overheads, which enhances the overall user experience by increasing the ease of implementation of the RIS in different environments.

    \item \textit{Compatibility with existing end-to-end channel estimation methods:} 
    The pre-configured phase shift design of the suggested plug-in RIS eliminates the necessity for decoupling two parts of the end-to-end channel, which alleviates the additional burden of channel decoupling on the participating terminals, enabling the easy deployment of fully passive RIS in mmWave communication systems. In other words, in the proposed plug-in RIS, obtaining the conventional end-to-end channel estimation is enough, and decoupling it into BS-RIS and RIS-UE channels is not required.

    \item \textit{Reliability analysis:} 
    We derive a closed-form expression for the upper bound of the average bit error rate (ABER) for the proposed plug-in RIS. Our analytical methodology employs the union-bound and pairwise error probability (PEP) to validate the precision of our simulation outcomes. Our results illustrate that these analytical derivations offer a tight upper bound on the ABER.
    
    \item \textit{Simulation insights:}
    We conduct extensive simulations to evaluate the performance of the proposed plug-in RIS in terms of {coverage performance}, ABER, achievable rate, energy efficiency, and detector complexity. Computer simulation outcomes indicate that the plug-in RIS exhibits promising {coverage performance}, showing close ABER and achievable rate results to the semi-passive design. Moreover, the plug-in RIS outperforms the benchmark designs in terms of energy efficiency and detector complexity, highlighting the superiority of our proposed approach, taking a step forward to the practical implementation of fully passive RIS. {Furthermore, analyzing received signal-to-interference-plus-noise-ratio (SINR) in an MU scenario reveals that plug-in RIS can also be adapted to serve more than a single UE without considerable performance loss compared to semi-passive RIS.}
    
\end{itemize}

The rest of this article is organized as follows. Section \ref{Sec: System,channel,signal model} outlines the system, channel, and signal models. The design of the proposed plug-in RIS is elaborated in Section \ref{Sec: Plug-in RIS Structure}. Analytical expression as an upper bound for ABER performance is derived in Section \ref{Sec: Theoretical analysis}. Simulation results are displayed in Sections \ref{Sec: Sim results} and \ref{Sec:Illustrative Results (MU)} for single-user and multi-user cases, respectively, followed by the conclusion of this study in Section \ref{Sec: Conclusion}.

\blfootnote{
\textit{Notation}: In this paper, bold lowercase and uppercase letters denote column vectors and matrices, respectively. The symbols $(.)^H$, $(.)^T$, $(.)^*$ $|.|$, $||.||$, and $\diag(.)$ stand for Hermitian, transpose, conjugate, absolute value, norm, and diagonalization, respectively. The operator $\mathbb{E}[.]$ shows the expected value, $\mathcal{R}(.)$ denotes the real component of a complex number, and operator $\circ$ represents the Hadamard (element-wise) product. The notation $\mathcal{CN}(\mu,\sigma^2)$ represents the complex Gaussian distribution with mean $\mu$ and variance $\sigma^2$, while $\mathcal{N}(\mu,\sigma^2)$ denotes the Gaussian distribution with mean $\mu$ and variance $\sigma^2$. The identity matrix of dimension $n$ is represented as $\mathbf{I}_n$, while the vector comprising all zeros and ones of size $n$ is denoted by $\mathbf{0}_n$ and $\mathbf{1}_n$, respectively. Additionally, $\mathcal{U}(a,b)$ indicates a uniform distribution parameterized by $a$ and $b$. The $Q$-function is denoted and defines as $Q(x) = \frac{1}{\sqrt{2 \pi}}\int_x^{\infty} \exp{(-\frac{u^2}{2})} du$. It is worth mentioning that in this paper, the variable $i$ is employed as a local variable, and its assigned value within each equation is applicable only to that equation.}

\vspace{-3em}

\section{System, Channel, and Signal Model}
\label{Sec: System,channel,signal model}
This section describes the system model for the mmWave massive MIMO communication system aided by the proposed plug-in RIS. Afterward, we present the utilized channel model and the foundational assumptions supporting our suggested design. Next, we elaborate on the signal model and introduce a suitable maximum likelihood (ML) detector. {Lastly, the proposed plug-in RIS is extended to an MU scenario considering inter-user interference and a baseband (BB) precoding stage to combat this interference.}

\vspace{-1em}

\subsection{System Model} 
\vspace{-0.5em}

As depicted in Fig. \ref{fig:SystemModel}, we investigate a P2P mmWave communication system aided by a set of sub-RIS planes with pre-configured fixed phase shifts.
Due to propagation loss challenges, the direct BS-UE link is assumed to be blocked, causing dead zones in the environment; therefore, the sub-RIS planes are plugged into the surrounding environment to assist UEs located in the dead zone, ensure constant connectivity and guarantee LOS links in the BS-RIS and RIS-UE channels. Specifically, a dead zone can be covered using one or more sub-RISs, each covering a distinct spatial segment of the dead zone with a fixed beam.
The set of sub-RISs constitutes the plug-in RIS entity. 
During implementation, each sub-RIS is assigned to a spatial segment in the dead zone; therefore, they are adjusted according to the angular information of BS and the spatial segments. Specifically, for each spatial segment, we consider fixed angular information that represents it, and since this angular information is constant, the sub-RISs are not required to be adjusted frequently.
During the transmission, based on the UE's location and with the aid of massive MIMO's beamforming gain at the BS, the transmitted beam footprint only activates the corresponding sub-RIS. Both the BS and UE are equipped with uniform rectangular arrays (URA); $N_t$ and $N_r$ denote the number of antenna elements at the BS and UE, respectively. To address severe path loss while minimizing the need for costly RF chains, analog beamformers/combiners are employed at the BS/UE. This approach allows a single RF chain at the BS and UE, supporting a single stream \cite{10188340}.
Our proposed plug-in RIS comprises $K$ sub-RISs, each containing $M_{k} \  (k = 1, \dots, K)$ passive reflecting elements.

\begin{figure}
    \centering
    \includegraphics[scale = 0.1]{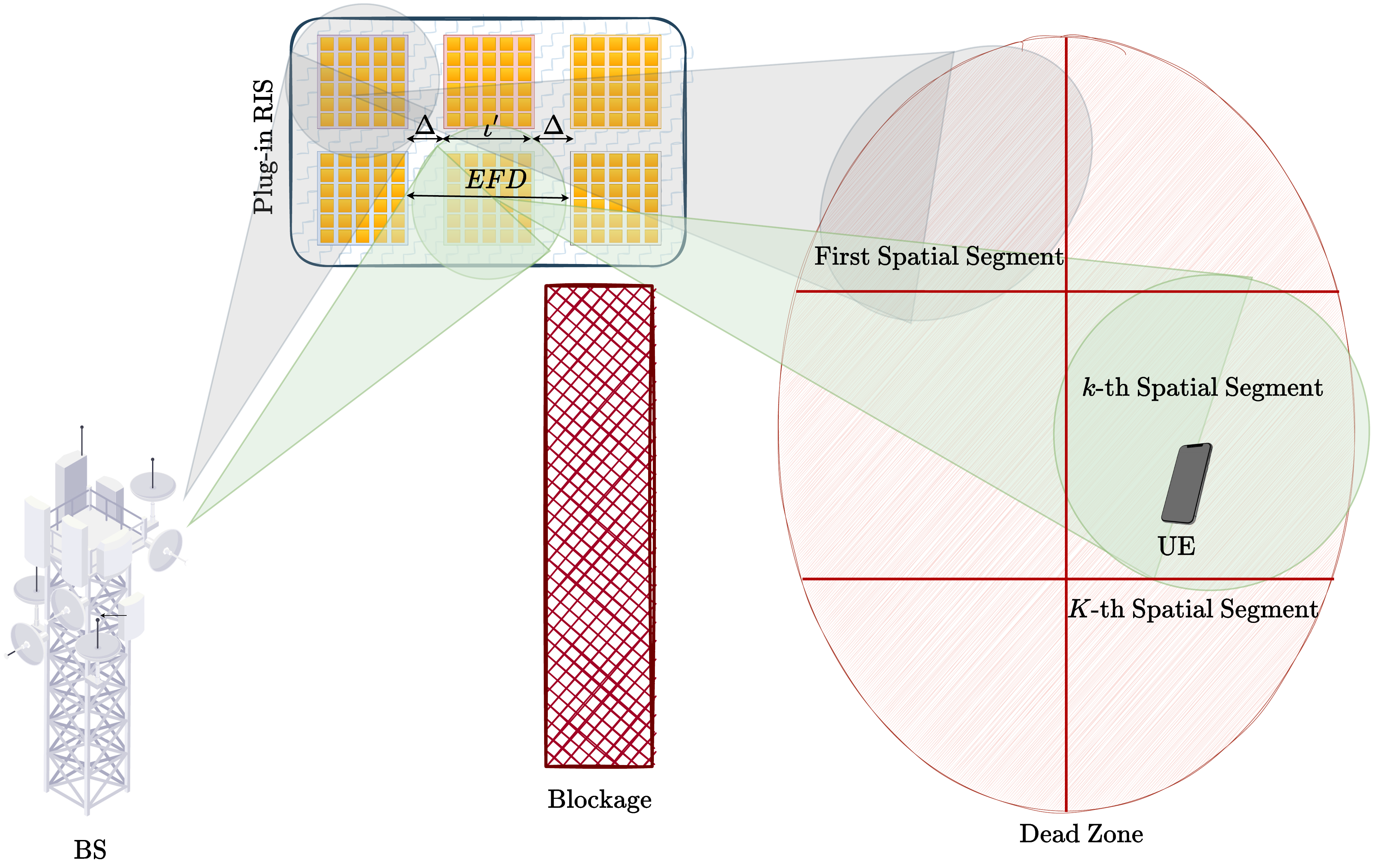}
    \caption{System model for proposed plug-in RIS.}
    \label{fig:SystemModel}
    \vspace{-1em}
\end{figure}

\vspace{-1em}
\subsection{Channel Model}
In this paper, we adopt the {three-dimensional (3D) geometry-based} channel model for the BS-RIS and RIS-UE channels as follows \cite{wang2020joint, wang2020intelligent, ying2020gmd, raeisi2022cluster, koc2020hybrid}:
\begin{equation}
\begin{split}
    \mathbf{G}_k = 
    & \sqrt{N_t M_k} \alpha_k \mathbf{a}_{r,ris}(\varphi_{r,ris}^k,\vartheta_{r,ris}^k) \mathbf{a}_t^H(\varphi_t,\vartheta_t),
\end{split}
\end{equation}
\begin{equation}
\begin{split}
    \mathbf{R}_k = 
    & \sqrt{M_k N_r} \beta_k \mathbf{a}_{r}(\varphi_{r},\vartheta_{r}) \mathbf{a}_{t,ris}^H(\varphi_{t,ris}^k,\vartheta_{t,ris}^k),
\end{split}
\end{equation}
where $\mathbf{G}_k \in \mathbb{C}^{M_{k} \times N_t}$ ($\mathbf{R}_k \in \mathbb{C}^{N_r \times M_{k}}$) is the channel between BS and $k$-th sub-RIS ($k$-th sub-RIS and UE).
The parameter $\alpha_k$ ($\beta_k$) is the LOS channel gain of $\mathbf{G}_k$ ($\mathbf{R}_k$) and follows a complex normal distribution $\mathcal{CN} (0, 10^{-0.1 PL(d)})$, with $PL(d)$ representing path loss and $d$ indicating the distance between the associated terminals. The array response (steering) vector is denoted as $\mathbf{a}(.)$, while $\varphi \in \mathcal{U}(a_{\varphi}, b_{\varphi})$ and $\vartheta \in \mathcal{U}(a_{\vartheta}, b_{\vartheta})$ stand for azimuth and elevation angles, respectively. 
Specifically, $\varphi_t$ ($\vartheta_t$) represents the azimuth (elevation) angle of departure (AoD) at the BS, and $\varphi_r$ ($\vartheta_r$) indicates the azimuth (elevation) angle of arrival (AoA) at the UE. {Besides}, $\varphi_{t,ris}^k$ ($\vartheta_{t,ris}^k$) shows azimuth (elevation) AoD associated with the $k$-th sub-RIS, and $\varphi_{r,ris}^k$ ($\vartheta_{r,ris}^k$) is the azimuth (elevation) AoA at the $k$-th sub-RIS.
Here, it is assumed that each channel is LOS dominated due to the presence of LOS component \cite{10176315}. Thus, non-LOS (NLOS) components are omitted, allowing for a focused investigation of the proposed plug-in RIS performance. The normalized antenna array response is defined as follows \cite{10188340, 10176315}:
\begin{equation}\label{eq: antenna array response}
    \mathbf{a}(\varphi, \vartheta) = \frac{1}{\sqrt{N}} [e^{j \mathbf{k}^T \mathbf{p}_0}, e^{j \mathbf{k}^T \mathbf{p}_1}, \dots, e^{j \mathbf{k}^T \mathbf{p}_{N-1}}]^T,
\end{equation}
where $N$ represents the total number of antenna elements/reflectors at the corresponding terminal, calculated as $N = N_x \times N_y$. Here, $N_x$ ($N_y$) refers to the number of antenna elements/reflectors along the $x$-axis ($y$-axis). 
The symbol $\mathbf{k}$ denotes the wave number, and $\mathbf{p}_n$ is associated with the structure of the antenna array, specifying the location of the $n$-th antenna element. We establish $n = n_y N_x + n_x$, where $0 \leq n_x \leq N_x - 1$ ($0 \leq n_y \leq N_y - 1$) represents the location of the antenna element/reflector along the $x$-axis ($y$-axis). For the URA configuration, the definitions of $\mathbf{k}$ and $\mathbf{p}_n$ are outlined as follows \cite{10188340, 10176315}:
\vspace{-0.5em}
\begin{equation}\label{eq: wave number}
    \mathbf{k} = \frac{2 \pi}{\lambda}[\sin{\vartheta}\cos{\varphi} \enspace \sin{\vartheta}\sin{\varphi}]^T,
    \vspace{-0.5em}
\end{equation}
\begin{equation}
    \mathbf{p}_n = [n_x \delta_x \enspace n_y \delta_y]^T,
    \vspace{-0.5em}
\end{equation}
where $\delta_x$ ($\delta_y$) represents the separation between adjacent antennas/reflectors along the $x$-axis ($y$-axis), and $\lambda$ denotes the signal wavelength. 
Note that to calculate the antenna array response for each terminal, we need to substitute its corresponding $\varphi$ and $\vartheta$ values into (\ref{eq: antenna array response}) and (\ref{eq: wave number}).
The effective end-to-end channel between the BS and UE through the $k$-th sub-RIS, denoted as $\mathbf{H}_{\text{eff},k} \in \mathbb{C}^{N_r \times N_t}$, can be defined as follows:
\begin{equation}\label{eq: effective channel}
    \mathbf{H}_{\text{eff},k} = \mathbf{R}_k \mathbf{\Psi}_k \mathbf{G}_k,
    \vspace{-0.5em}
\end{equation}
where $\mathbf{\Psi}_k \in \mathbb{C}^{M_{k} \times M_{k}}$ is the $k$-th sub-RIS phase shift matrix and {is defined as}
\begin{equation}
    \mathbf{\Psi}_k \triangleq \diag{(e^{j \psi_1}, e^{j \psi_2}, \dots, e^{j \psi_{M_{k}}})},
    \vspace{-0.5em}
\end{equation}
where $\psi_i$ is the phase shift associated with the $i$-th reflector of the corresponding sub-RIS. 
As mentioned earlier, the sub-RISs can be adjusted using the location of BS and dead zones. Lemma \ref{Lem: RIS phase shift} gives a mathematical expression for sub-RISs phase shift adjustment.

\begin{lemma}
\label{Lem: RIS phase shift}
    The phase shift matrix for the $k$-th sub-RIS can be calculated as follows:
    \begin{equation}
        \mathbf{\Psi}_k = \sqrt{M_{k}} \diag \Big( \mathbf{f}_{b,k} \circ {\mathbf{a}_{r,ris}^*}(\varphi_{r,ris}^k,\vartheta_{r,ris}^k) \Big),
    \end{equation}
    where $\mathbf{f}_{b,k}$ represents the fixed beam formed by the $k$-th sub-RIS with azimuth and elevation AoD as $\phi_{b,k}$ and $\theta_{b,k}$, respectively. $\mathbf{f}_{b,k}$ is expressed as
    \begin{equation*}
        \mathbf{f}_{b,k} = [e^{\mathbf{k}^T \mathbf{p}_0}, e^{\mathbf{k}^T \mathbf{p}_1}, \dots, e^{\mathbf{k}^T \mathbf{p}_{M_{k}-1}}],
    \end{equation*}
    where $\mathbf{k}$ is calculated as (\ref{eq: wave number}) by substituting $\varphi = \phi_{b,k}$ and $\vartheta = \theta_{b,k}$.
\end{lemma}
 \begin{proof}
     See \cite[Section III-C]{ying2020gmd}.
 \end{proof}
 
\vspace{-2em}
\subsection{Signal Model}

As mentioned earlier, the dead zone is divided into distinct spatial segments based on its geometric shape. Given that the dead zone's location remains fixed, the spatial segments also maintain constant positions. Consequently, each spatial segment is defined by a pair of unchanging azimuth and elevation angular values. Thus, sub-RISs are adjusted according to the constant angular information of BS and spatial segments in order to cover that segment.
{For starting transmission, end-to-end CSI and the spatial segment of the UE should be obtained.}
Since the BS and sub-RISs have fixed positions, the angular information at the BS is constant and known to the BS; therefore, after transmitting pilot signals, the BS can easily find the spatial segment in which the UE is located and its associated sub-RIS. 
Furthermore, by estimating end-to-end channel, both the BS and UE can adjust their analog beamformers/combiners accordingly. 
The effective end-to-end channel can be described as (\ref{eq: effective channel}).
Therefore, the received signal after passing through the combiner is represented as
\begin{equation}
    y = \sqrt{P} G_t G_r \mathbf{f}_r^H \mathbf{H}_{\text{eff},k} \mathbf{f}_t s + \mathbf{f}_r^H \mathbf{n},
    \vspace{-0.5em}
\end{equation}
where $P$ stands for the transmit power, $G_i = g_e + 10 \log_{10}(N_i)$, wherein $i \in \{t, r \}$, represents the array gain, with $g_e$ denoting the gain of each antenna element \cite{hannan1964element, sarieddeen2019terahertz}. The symbol $s$ corresponds to the transmitted $\mathcal{M}$-ary symbol, and $\mathbf{n} \in \mathbb{C}^{N_r \times 1} \sim \mathcal{C}\mathcal{N}(0,\sigma^2)$ is the additive noise component. Moreover, $\mathbf{f}_t \in \mathbb{C}^{N_t \times 1}$ denotes the analog beamformer at the BS, while $\mathbf{f}_r \in \mathbb{C}^{N_r \times 1}$ represents the analog combiner at the UE. For simplification, continuous phase shifters are assumed at the BS and UE; thus, analog beamformer/combiner can be adjusted as $\mathbf{f}_t = \mathbf{a}_t(\varphi_t, \vartheta_t)$ and $\mathbf{f}_r = \mathbf{a}_r(\varphi_r, \vartheta_r)$.
Utilizing the ML detector at the UE, the estimated symbol is expressed as
\begin{equation}\label{eq:Detector}
    \hat{s} = \argmin_{\forall s \in \mathcal{S}} |y - \sqrt{P} G_t G_r \mathbf{f}_r^H \mathbf{H}_{\text{eff},k} \mathbf{f}_t s|^2,
\end{equation}
where $\mathcal{S}$ is the set of $\mathcal{M}$-ary constellation points.

\vspace{-1em}
{
\subsection{MU Plug-in RIS}
This subsection investigates exploiting of plug-in RIS in MU scenarios to serve $J$ single antenna UEs. We introduce the channel model for the MU-plug-in RIS-assisted system, then present the signal model and ultimately derive the SINR term.
        
The overall effective channel matrix, denoted as $\mathbf{H} \in \mathbb{C}^{J \times N_t}$, can be obtained as follows:
\begin{equation}
    \mathbf{H} = \Big[ \mathbf{h}_{\textrm{eff},1}^T, \dots, \mathbf{h}_{\textrm{eff},J}^T \Big]^T,
\end{equation}
where $\mathbf{h}_{\textrm{eff},j} \in \mathbb{C}^{1 \times N_t}$, ($j \in \{ 1, \dots, J \}$), represents the end-to-end effective channel between the BS and the $j$-th UE. This can be computed according to (\ref{eq: effective channel}), considering single antenna UEs, i.e., $N_r = 1$. It is worth mentioning that for clarity, we do not include sub-RIS index $k$ in this subsection; however, each UE initiates the communication through a sub-RIS, which can be the same or different from other UEs.
The effective channel seen from the BS's BB stage can be represented as $\mathbfcal{H} = \mathbf{H} \mathbf{F_t}$, where $\mathbf{F_t} = \left[ \mathbf{f}_{t,1}, \dots, \mathbf{f}_{t,J} \right]$ is the transmit analog active beamforming matrix. Here, $\mathbf{f}_{t,j} \in \mathbb{C}^{N_t \times 1}$ denotes the transmit analog beamforming vector corresponding to the $j$-th UE.
Accordingly, the BB precoding matrix is calculated as $\mathbf{B} = \varepsilon \mathbf{W} \mathbfcal{H}^H$ where $\mathbf{W} = (\mathbfcal{H}^H \mathbfcal{H} + J \alpha \mathbf{I}_J)^{-1}$ \cite{8891298}. Here, $\alpha = \sigma^2/P$ is the regularization parameter, and $\varepsilon = \sqrt{P_T / \Trace \big( \mathbfcal{H} \mathbf{W}^2 \mathbfcal{H}^H \big)}$ is the normalization scalar \cite{8891298}. Therefore, the received signal in the $j$-th UE is given by:
\begin{equation}
\begin{split}
    y_k = & \underbrace{\varepsilon \mathbf{H}_{\textrm{eff},j} \mathbf{F_t} \mathbf{W} \mathbf{F}^H \mathbf{H}_{\textrm{eff},j}^H s_j}_{\textrm{Desired Signal}} \\ 
    & + \ \underbrace{\varepsilon \sum_{q = 1}^J \mathbf{H}_{\textrm{eff},j} \mathbf{F_t} \mathbf{W} \mathbf{F}^H \mathbf{H}_{\textrm{eff},q}^H s_q}_{\textrm{Interference}} \ + \ \underbrace{\mathbf{n}_j}_{\textrm{Noise}}.
\end{split}
\end{equation}
Finally, the SINR can be derived as follows:
\begin{equation}
    \textrm{SINR} = \frac{|\mathbf{H}_{\textrm{eff},j} \mathbf{F_t} \mathbf{W} \mathbf{F}^H \mathbf{H}_{\textrm{eff},j}^H|^2}{|| \mathbf{H}_{\textrm{eff},j} \mathbf{F_t} \mathbf{W} \mathbfcal{H}_{[j]}||^2 \ + \ \mathbf{n}_j},
\end{equation}
where $\mathbfcal{H}_{[j]} = [\mathbf{h}_{\textrm{eff},1}^T, \dots, \mathbf{h}_{\textrm{eff},j - 1}^T, \mathbf{h}_{\textrm{eff},j + 1}^T, \dots, \mathbf{h}_{\textrm{eff},J}^T]^T \mathbf{F_t}$ is the reduced-sized effective interference channel matrix.
}

\vspace{-1em}
\section{Proposed Plug-in RIS Structure}
\label{Sec: Plug-in RIS Structure}
\vspace{-0.5em}

This section illustrates the structure of the proposed plug-in RIS. 
In the plug-in RIS structure, sub-RIS units are placed with appropriate spacing, each pre-configured to cover a specific spatial segment within a {known dead zone}. 
Based on environmental blockage and path loss conditions, the sub-RIS units can be positioned on either the BS or UE side. The spacing between sub-RIS units should increase as they are deployed farther from the BS. We subsequently formulate the minimum spacing between sub-RIS units. 

Maintaining an appropriate spacing between sub-RIS units is crucial for reducing power leakage to other segments within the dead zone, which is particularly beneficial in MU scenarios as it leads to decreased interference.
Leveraging the beamforming gain offered by the large antenna array at the BS, forming a narrow beam precisely directed to a specific spatial area becomes feasible. This beam's coverage area, termed footprint, is determined by the distance from the BS and the emitted signal's beamwidth. Exploiting this property of large arrays enables the BS to illuminate only the corresponding sub-RIS based on the UE's location. 
In Fig. \ref{fig:SystemModel}, an illustration of adjacent sub-RISs separated by $\Delta$ is shown.
To establish a closed-form expression for $\Delta$, we initially define the length of each side of the sub-RIS along the $i$-axis  (where $i \in \{ x,y \}$) as $\iota_i = M_i \delta_i$.
Note that $M_i$ represents the number of reflectors across the $i$-axis, and for simplicity in notation, we eliminated the sub-RIS's index.
The footprint of the received beam should exclusively cover the intended sub-RIS. Assuming uniform spacing between sub-RISs, the effective footprint diameter (EFD) can be computed as $\textrm{EFD} = 2\Delta + \iota$. Consequently, $\Delta$ can be derived as follows:
\begin{equation}\label{eq: sub-RIS spacing}
    \Delta = \frac{\textrm{EFD} - \iota}{2}.
\end{equation}

In order to determine the EFD, we consider half-power beamwidth (HPBW) that encompasses an effective portion of the radiated power \cite{gopi2020intelligent}. It is worth mentioning that an antenna can differentiate between two adjacent sources if the angular separation between them exceeds half of the first-null beamwidth (FNBW), which is approximately equivalent to HPBW, i.e., $\textrm{HPBW} \approx \frac{\textrm{FNBW}}{2}$ \cite{balanis2015antenna}. The HPBW of an antenna array along the $i$-axis ($i \in \{ x, y \}$) can be calculated using the following formula \cite{van2002optimum}:
\begin{equation}
    \textrm{HPBW} \approx 0.891 \frac{\lambda}{N_i \delta_i}.
\end{equation}

\begin{lemma}
    The EFD of the transmitted beam from the Tx at a distance of $d$ can be computed as follows:
    \begin{equation}\label{eq: EFD exact}
        \textrm{EFD} = 2d \times \frac{\sin{(\frac{\textrm{HPBW}}{2})}}{\cos{(\frac{\textrm{HPBW}}{2} + \theta_0)}},
    \end{equation}
    where $\theta_0$ represents the angle of the incident signal at the RIS with respect to the RIS broadside direction.
\end{lemma}

\begin{proof}
    See \cite[Appendix B]{9386246}.
\end{proof}

\begin{figure}
    \centering
    \includegraphics[scale = 0.37]{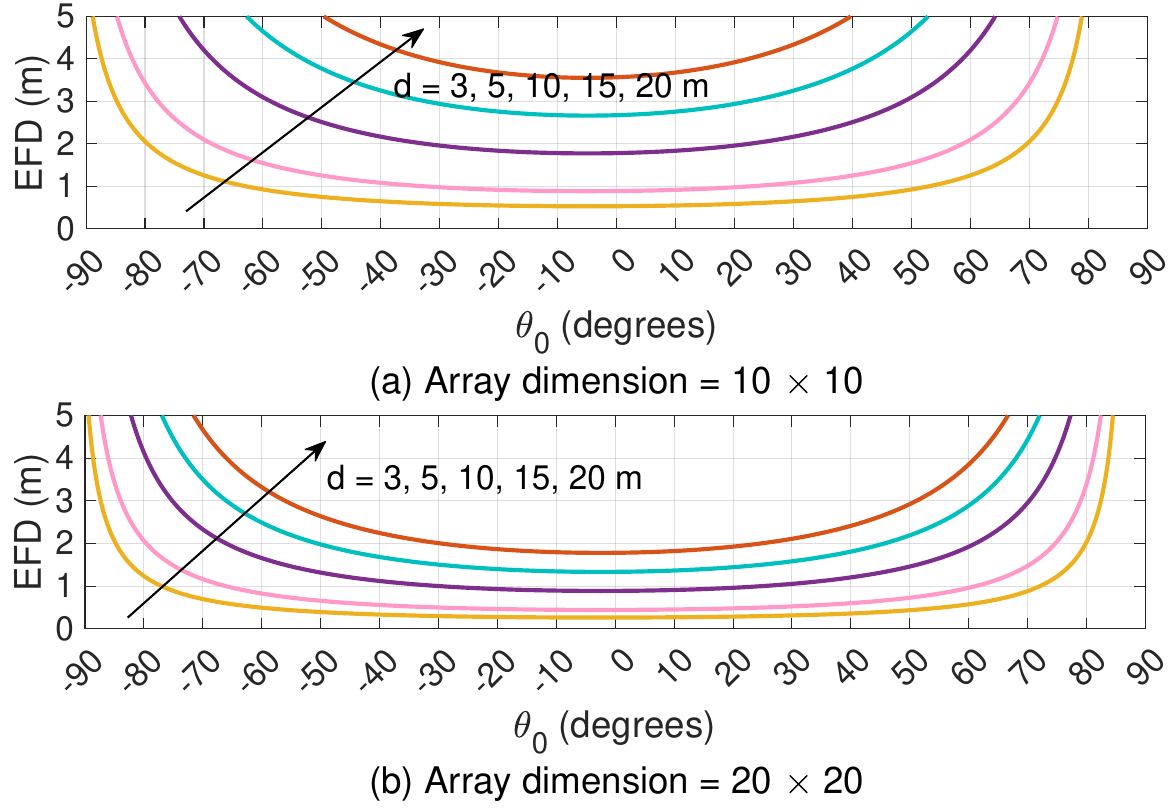}
    \caption{EFD versus different amounts of incident angle.}
    \label{fig:EFD_IncidentAngle}
    \vspace{-2em}
\end{figure}

\vspace{-1em}
\begin{corollary}\label{Cor: EFD Approximation}
In the case where $\theta_0 = 0$, indicating that the received beam at the RIS is perpendicular to the RIS plane, the EFD can be approximated as
\begin{equation}\label{eq: EFD approximation}
    \textrm{EFD} \approx \textrm{HPBW} \times d.
\end{equation}
\end{corollary}
\begin{proof}
See Appendix \ref{Apndx: EFD Aprox Proof}.
\end{proof}
Fig. \ref{fig:EFD_IncidentAngle} illustrates the precise EFD for various values of $d$ and $\theta_0$. Notably, in Figs. \ref{fig:EFD_IncidentAngle}(a) and (b), a distinct pattern emerges: within a limited range centered around $\theta_0 = 0^\circ$, the EFD remains constant; the extent of this interval diminishes with higher values of $d$. Conversely, as it is illustrated in Fig. \ref{fig:EFD_IncidentAngle}(b), adopting a larger antenna array leads to expanding the interval.

Substituting EFD into (\ref{eq: sub-RIS spacing}), sub-RISs inter-spacing, i.e., $\Delta$, can be calculated as
\begin{equation}
    \Delta = \frac{\textrm{HPBW} \times d - \iota}{2}.
\end{equation}
Notably, $\Delta$ defines the minimum distance that needs to be maintained between two neighboring sub-RIS units to mitigate power leakage to the other spatial segments.

\vspace{-1em}
\section{Theoretical Analysis}
\label{Sec: Theoretical analysis}
In this section, we begin by presenting a mathematical expression to validate the accuracy of our ABER simulation results. Following that, we provide expressions for computing the detector's complexity for both plug-in and semi-passive RIS designs.

\vspace{-1em}
\subsection{ABER Theoretical Upper Bound}

Obtaining ABER involves calculating PEP, which denotes the likelihood of detecting $\hat{s}$ when transmitting $s^*$. Initially, we consider the channel as a known entity for both cascaded channels. This allows us to express the conditional PEP (CPEP) as follows:
\begin{equation}\label{eq: 1st equation CPEP}
\begin{split}
    \mathcal{P}( s^*  \rightarrow   \hat{s}  | \alpha_k, \beta_k ) 
    = & \mathcal{P} (|y - \sqrt{P} G_t G_r \mathbf{f}_r^H \mathbf{H}_{\text{eff},k} \mathbf{f}_t s^*|^2 \\
    & > |y - \sqrt{P}G_t G_r \mathbf{f}_r^H \mathbf{H}_{\text{eff},k} \mathbf{f}_t \hat{s}|^2).
\end{split}
\end{equation}

\begin{lemma} \label{Lem: CPEP}
    The mathematical expression for CPEP is as
    \begin{equation}\label{eq: CPEP}
    \begin{split}
        \mathcal{P} ( s^*  \rightarrow  & \hat{s} |\alpha_k,\beta_k) \\
        & = Q\Big(\frac{\sqrt{P}G_t G_r |\mathbf{f}_r^H \mathbf{H}_{\text{eff},k} \mathbf{f}_t(s^* - \hat{s})|^2}{\sqrt{2}\sigma || \mathbf{f}_r \mathbf{f}_r^H \mathbf{H}_{\text{eff},k} \mathbf{f}_t (s^* - \hat{s})||}\Big).
    \end{split}
    \end{equation}
\end{lemma}
\begin{proof}
    See Appendix \ref{Apndx: CPEP proof}.
\end{proof}
Lemma \ref{Lem: CPEP} provides a closed-form expression for the CPEP. To compute the unconditional PEP (UPEP), we take the expectation of CPEP over a sufficient number of channel realizations, denoted as $\mathrm{UPEP} = \mathbb{E} [\mathrm{CPEP}]$. Furthermore, to establish an upper bound for the ABER, we utilize the union-bound approach as follows:
\begin{equation}
    \textrm{ABER} \leq \frac{1}{\eta \mathcal{M}} \sum_{s^*} \sum_{\hat{s}} E_b (  s^*  \rightarrow  \hat{s}  ) \textrm{UPEP},
\end{equation}
where $E_b (  s^*  \rightarrow  \hat{s}  )$ is the total number of erroneous bits and $\mathcal{M}$ is the order of $\mathcal{M}$-ary constellation.

\vspace{-1em}
\subsection{Detector Complexity}
\vspace{-0.5em}

This subsection discusses the computation of detector complexity for the proposed plug-in RIS and semi-passive RIS designs.
Both designs employ the same ML detector given in (\ref{eq:Detector}), but the main difference lies in obtaining $\mathbf{H}_{\text{eff}}$. In the semi-passive RIS design, $\mathbf{G}$ and $\mathbf{R}$ are computed separately, and the phase shift matrix $\mathbf{\Psi}$ is also calculated by the RIS in each transmission period. The detector then uses these matrices to calculate $\mathbf{H}_{\text{eff}}$, affecting the detector's complexity.
However, in the plug-in RIS design, $\mathbf{H}_{\text{eff}}$ is directly acquired without the need for separate calculations of $\mathbf{G}$, $\mathbf{R}$, and $\mathbf{\Psi}$ due to the pre-adjusted sub-RIS configuration. Considering this point, we compute each RIS systems's detector complexity as follows. In this subsection, for the sake of notation simplicity, we show the number of RIS/sub-RIS elements engaged in the communication with $M$; hence, the index of sub-RISs has been ignored.

\begin{table}[t!]
  \scriptsize
  \centering
  \caption{Simulation parameters.}
  \vspace{-1em}
  \arrayrulecolor{black}
  \begin{tabular}{c}
  \begin{tabular}[t]{ p{0.05cm}  p{0.5cm}  p{3.4cm}  p{2.5cm} }
  \toprule
  \multicolumn{2}{c}{\textbf{Parameter}} &  \textbf{Description}  & \textbf{Value}  \\
  \toprule
  \multicolumn{2}{c}{$f_c$}  & Carrier frequency  & $28$ GHz  \\ [1 pt]
  \midrule
  \multicolumn{2}{c}{$B$}    & Bandwidth & $100$ MHz   \\ [1 pt]
  \midrule
  \multicolumn{2}{c}{$N_t$}  & BS antenna array size & $10 \times 10$    \\ [1 pt]
  \midrule
  \multicolumn{2}{c}{$N_r$}  & UE antenna array size  & $1$    \\ [1 pt]
  \midrule
  \multicolumn{2}{c}{$\delta_x = \delta_y$} & Antenna element separation & $\lambda / 2$    \\ [1 pt]
  \midrule
  \multicolumn{2}{c}{Noise PSD}    & Noise power spectral density & $-174$ dBm    \\
  \midrule
  \multicolumn{2}{c}{$g_e$}  & Antenna element gain & $0$ dBi   \\ 
  \midrule
  \multicolumn{2}{c}{$\varphi_t$ ($\varphi_{r,ris}$)} & Azimuth Tx AoD (RIS AoA) & $\mathcal{U}[-\pi, \pi]$   \\ [1 pt]
  \midrule
  \multicolumn{2}{c}{$\varphi_{t,ris}$} & Azimuth RIS AoD  & $\mathcal{U}[-\pi, \pi]$   \\ [1 pt]
  \midrule
  \multicolumn{2}{c}{$\vartheta_t$ ($\vartheta_{r,ris}$)} & Elevation Tx AoD (RIS AoA) & $\mathcal{U}[-\pi/3, \pi/3]$   \\ [1 pt]
  \midrule
  \multicolumn{2}{c}{$\vartheta_{t,ris}$}  & Elevation RIS AoD & $\mathcal{U}[-\pi/16, \pi/16]$\\
  \midrule
  \multicolumn{2}{c}{$(\phi_{b,k},\theta_{b,k})$}  & Azimuth and elevation AoD for a two sub-RIS system& $\{ (\frac{\pi}{2}, \frac{\pi}{32})$,$ (-\frac{\pi}{2}, \frac{\pi}{32}) \}$\\
  \midrule
  \multicolumn{2}{c}{$(\phi_{b,k},\theta_{b,k})$}  & Azimuth and elevation AoD for a four sub-RIS system& $\{ (\frac{\pi}{4}, \frac{\pi}{32})$,$(\frac{3\pi}{4}, \frac{\pi}{32})$, $(-\frac{\pi}{4}, \frac{\pi}{32})$,$(-\frac{3\pi}{4}, \frac{\pi}{32})\}$\\
  \midrule
  \multicolumn{2}{c}{{$P_{\text{controller}}$}}  & {RIS controller power consumption} & {$1.5$ W \cite{9551980}}\\
  \midrule
  \multicolumn{2}{c}{$P_{\text{PA}}$}  & Power amplifier power consumption & $20$ mW \cite{7370753}\\
  \midrule
  \multicolumn{2}{c}{$N_{\text{rf}}$}  & Number of RF chain & $1$\\
  \midrule
  \multicolumn{2}{c}{$M_{\text{active}}$}  & Number of RIS active elements & $0.08 \times M$ \cite{9685434}\\
  \midrule
  \multicolumn{2}{c}{$P_{\text{PS}}$}  & Phase shifter power consumption & $30$ mW \cite{7370753}\\
  \midrule
  \multicolumn{2}{c}{$P_{\text{RF-chain}}$}  & RF chain power consumption & $40$ mW \cite{7876856, 7370753}\\
  \midrule
  \multicolumn{2}{c}{$P_{\text{BB}}$}  & Baseband processor power consumption & $200$ mW \cite{7876856, 7370753}\\
  \midrule
  \multicolumn{2}{c}{$P_{\text{LNA}}$}  & Low noise amplifier power consumption & $20$ mW \cite{7876856, 7370753}\\
  \midrule
  \multicolumn{2}{c}{$P_{\text{PA\_RIS}}$}  & RIS element power consumption & $10$ mW \cite{9370097}\\
  \midrule
  \multicolumn{2}{c}{$FOM_W$}  & Walden’s figure-of-merit & $46.1$ fJ/conversion-step \cite{9370097}\\
  \midrule
  \multirow{3}*{\rotatebox[origin=c]{90}{Indoors}} & $(a,b)$ & Path loss parameters & $(32.4,1.73)$ \cite{etsitr138901v17}  \\
       & $d_{BR}$ & Distance between BS and RIS & $2.5 \ m$ \\
       & $d_{RU}$ & Distance between RIS and UE & $10 \ m$ \\[1 pt]
  \midrule 
  \multirow{3}*{\rotatebox[origin=c]{90}{Outdoors}} & $(a,b)$ & Path loss parameters & $(32.4,2.1)$ \cite{etsitr138901v17}  \\
       & $d_{BR}$ & Distance between BS and RIS & $20 \ m$ \\
       & $d_{RU}$ & Distance between RIS and UE & $10 \ m$ \\
  [3pt]
  \bottomrule
\end{tabular}
\end{tabular}
  \label{tab:System parameters}
  \vspace{-3em}
\end{table}

\begin{enumerate}
    \item \textit{\textbf{Semi-passive RIS:}} To compute $\mathbf{H}_{\text{eff}}$, the detector calculates $\mathbf{\Psi}\mathbf{G}$, involving $\sim M^2 N_t$ operations. Subsequently, $\mathbf{R}\mathbf{\Psi}\mathbf{G}$ demands $\sim N_r M N_t$ operations. Thus, for computing $\mathbf{H}_{\text{eff}}$, $\sim M^2 N_t + N_r M N_t$ operations are needed. Moreover, $\mathbf{H}_{\text{eff}} \mathbf{f}_t$ entails $\sim N_r N_t$ operations, and obtaining $\mathbf{f}_r^H \mathbf{H}_{\text{eff}} \mathbf{f}_t$ takes $\sim N_r$ operations. With this procedure repeated $\mathcal{M}$ times, the detector's complexity is $\mathcal{O}(\mathcal{M}(M^2 N_t + N_r M N_t + N_t N_r + N_r))$. However, by neglecting the lower-order terms $N_t N_r + N_r$, the complexity of the semi-passive RIS detector can be simplified to $\sim \mathcal{O}(\mathcal{M} M N_t(M + N_r))$.

    \item \textit{\textbf{Plug-in RIS:}} In the plug-in RIS design, the detector has direct access to $\mathbf{H}_{\text{eff}}$ as explained earlier. Thus, in the initial step, it computes $\mathbf{H}_{\text{eff}} \mathbf{f}_t$, requiring $\sim N_r N_t$ operations. Following this, to calculate $\mathbf{f}_r^H \mathbf{H}_{\text{eff}} \mathbf{f}_t$, the detector performs $\sim N_r$ operations. This process is repeated $\mathcal{M}$ times, resulting in a detector complexity of order $\mathcal{O}(\mathcal{M}(N_r N_t + N_r))$. By omitting the lower-order term $N_r$, the plug-in RIS complexity can be simplified to $\sim \mathcal{O}(\mathcal{M} N_r N_t)$. 
\end{enumerate}

The complexity analysis conducted for both the semi-passive and plug-in RIS detectors demonstrates that the semi-passive RIS detector is more computationally complex, which is numerically shown in Section \ref{subsec: SIM RES-Detector complexity}.

\vspace{-1em}
\section{Illustrative Results - Single User Case}
\label{Sec: Sim results}
\vspace{-0.5em}

This section evaluates the proposed plug-in RIS {for single-user scenarios} using {five} performance metrics: {received signal-to-noise-ratio (SNR) for coverage performance}, ABER, achievable rate, energy efficiency (EE), and detector complexity.
We also validate the ABER performance of the plug-in RIS via the theoretical upper bound. The assessment is conducted in practical scenarios, both indoors (office) and outdoors (street canyon), implementing RIS at the BS and UE sides, respectively.
The effectiveness of the plug-in RIS is compared with {four} benchmarks, which are:
\begin{itemize}
    \item \textbf{Benchmark 1}: \textit{Semi-passive RIS} \cite{9370097, 9685434}. This structure has been proposed in the literature as a practical approach that eliminates the need for cascaded channel decoupling \cite{9685434}. Besides, semi-passive RIS can perform online (real-time) configuration based on the current channel characteristics.

    \item \textbf{Benchmark 2}: \textit{Blind RIS} \cite{8801961}\footnote{In \cite{8801961}, the blind RIS suggests setting all RIS elements' phase shifts to zero. However, to enhance the performance of the blind RIS in the mmWave communication systems, we set random phase shift adjustments to increase diversity.}. In the blind RIS configuration, the phase adjustments occur fully randomly; accordingly, there is no necessity for cascaded channel decoupling or control link establishment, like the proposed plug-in RIS. 

    \item \textbf{Benchmark 3}: \textit{Amplitude and forward (AF) relay} \cite{10183245}. AF relay is a traditional method to extend the coverage to the dead zones.

    \item \textbf{Benchmark 4}: \textit{Codebook-based RIS} \cite{9952197, 10097454}. This scheme aims to reduce the pilot overhead and streamline the implementation process compared to conventional RIS schemes.
\end{itemize}


\vspace{-2em}
\subsection{Simulation Setup}
\vspace{-0.5em}
In this subsection, we describe two wireless communication scenarios to evaluate the performance of the proposed plug-in RIS in enhancing the coverage. These scenarios have also presented in \cite{etsigrris001} as practical options for RIS implementation. 

\begin{figure}
     \centering
     \begin{subfigure}[b]{0.47\columnwidth}
         \centering
         \includegraphics[width=\columnwidth]{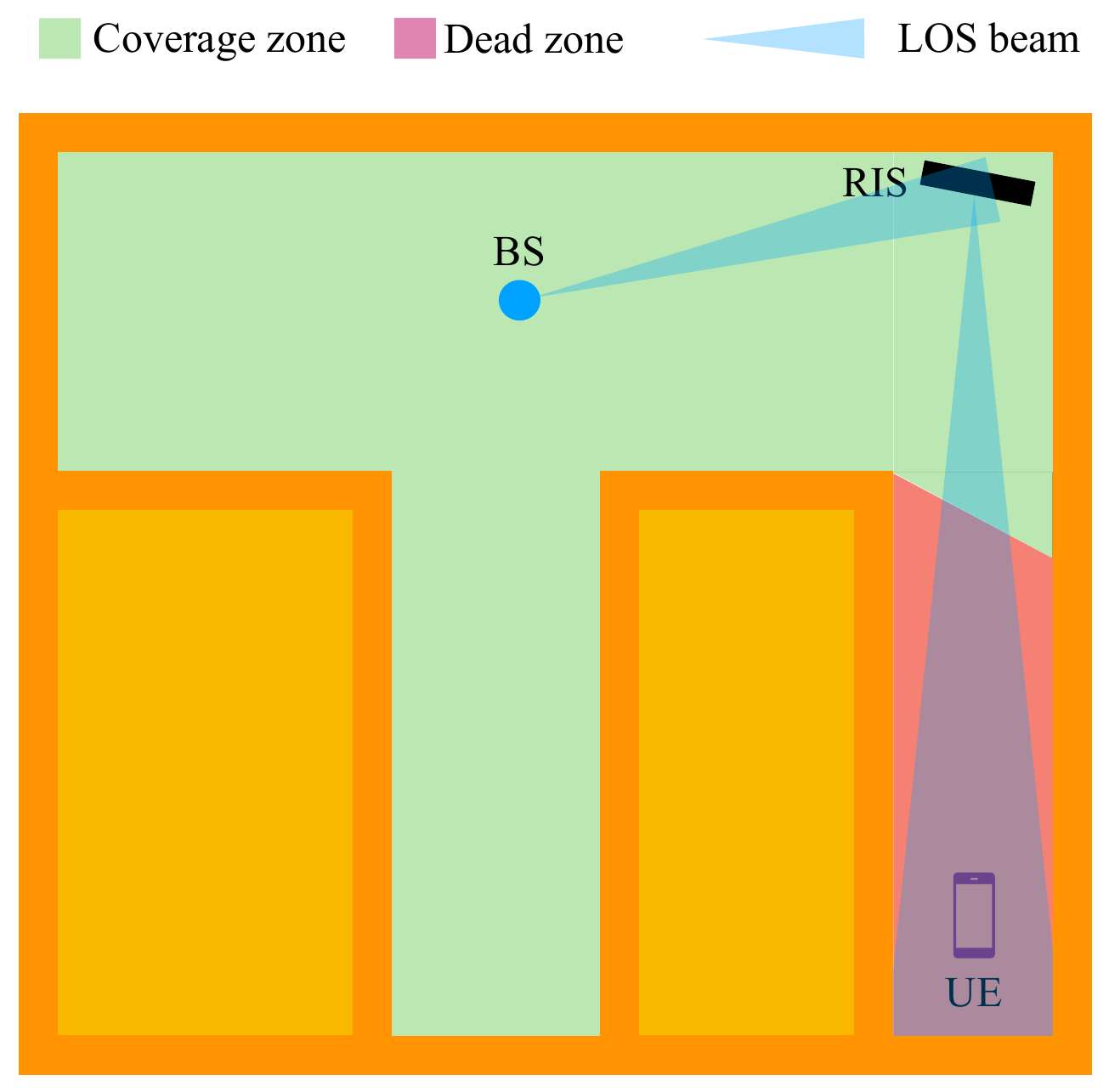}
         \caption{}
         \label{fig:Indoor Office Scenario}
     \end{subfigure}
     \hfill
     \begin{subfigure}[b]{0.47\columnwidth}
         \centering
         \includegraphics[width=\columnwidth]{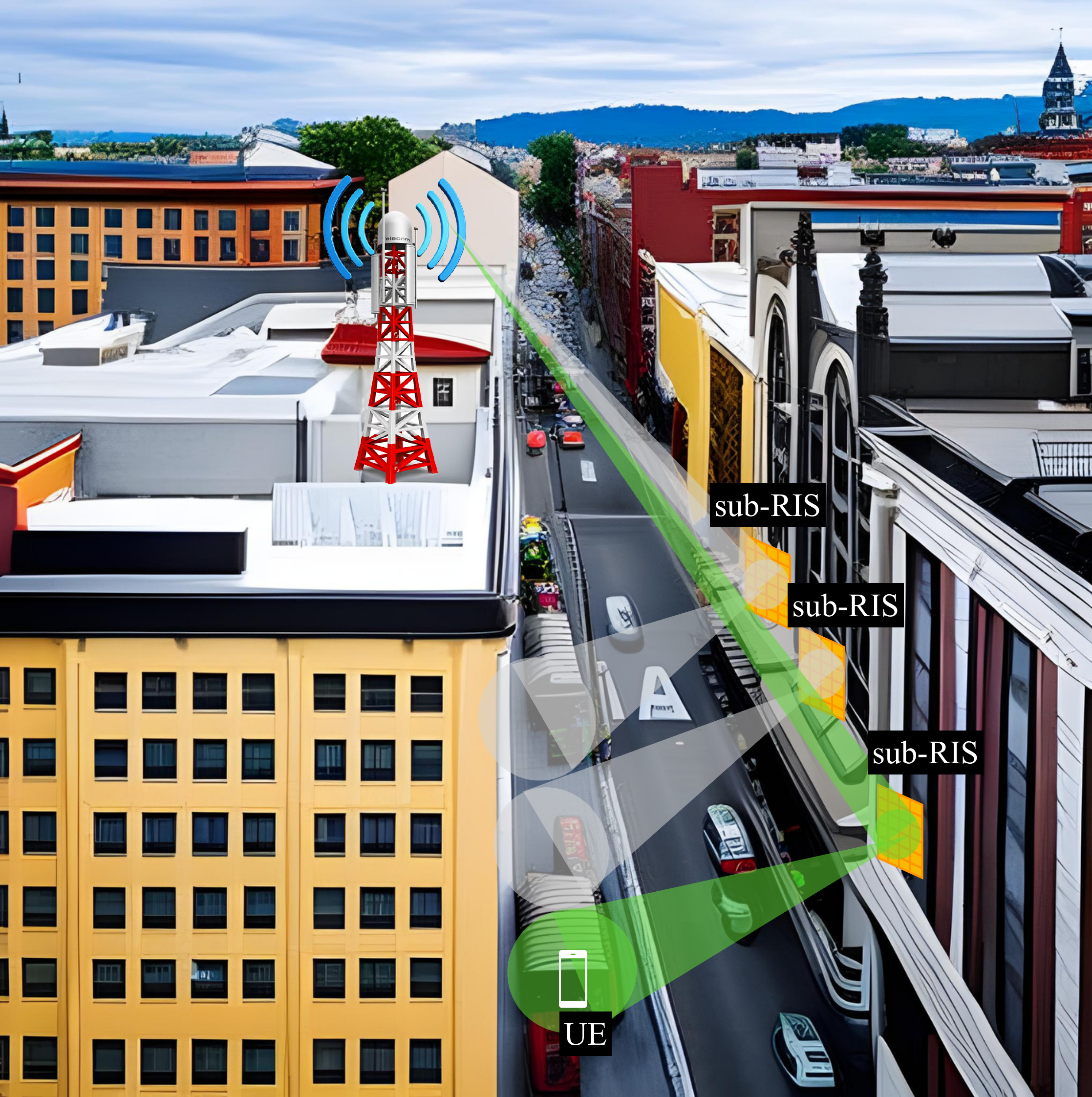}
         \caption{}
         \label{fig:Street Canyon Scenario}
     \end{subfigure}
     \hfill
     \caption{Implemented scenarios; (a) Indoor office for BS-side RIS, (b) Street canyon scenario for UE-side RIS.}
     \label{fig:scenarios}
     \vspace{-2em}
\end{figure}

\begin{enumerate}
    \item \textit{\textbf{Indoor Office (BS-side RIS):}} The first scenario involves an indoor environment with two parallel hallways, as illustrated in Fig. \ref{fig:scenarios}(a). In practical situations, the BS is optimally installed to cover maximum areas, but due to mmWave vulnerability to blockage, certain spots and zones become challenging to cover, for instance, one of the parallel hallways in this scenario. Instead of deploying additional costly and inefficient BS, the RIS can be implemented to extend coverage to these dead zones. Here, the distance between the BS (RIS) and the RIS (UE) is $2.5 \ m$ ($10 \ m$).

    \item \textit{\textbf{Street Canyon (UE-side RIS):}} For the second scenario, we consider a street canyon environment, shown in Fig. \ref{fig:scenarios}(b), where the BS is installed on the rooftop of one building on one side of the street. UEs positioned along the same side of the street experience signal deficiency owing to their placement within the dead zone. While an alternative could involve deploying another BS on top of a building on the opposite side of the street, adopting RIS proves to be a more cost/energy-effective way to cover the dead zone. In this scenario, we assume that the distance between the BS (RIS) and the RIS (UE) is $20 \ m$ ($10 \ m$).
    \vspace{-0.5em}
\end{enumerate}


The carrier frequency is designated as $28$ GHz, and the path loss model is employed as follows \cite{etsitr138901v17}:
\vspace{-0.5em}
\begin{equation}
    PL (d) = a + 10 b \log_{10}{(d)} + 20\log_{10}{(f_c)},
    \vspace{-0.5em}
\end{equation}
where the values of $a$ and $b$ are given in Table \ref{tab:System parameters}. The analysis assumes a noise power spectral density (PSD) of $-174$ dBm/Hz, with a bandwidth ($B$) of $100$ MHz \cite{10176315}. This results in a noise power calculation of $\sigma^2 = -94$ dBm \cite{10176315}. The antenna gain of each element is considered to be $g_e = 0$ dBi, signifying the utilization of omnidirectional antennas for every element. The assumed values for system parameters are summarized in Table \ref{tab:System parameters}.

\begin{figure}
    \centering
    \includegraphics[scale = 0.25]{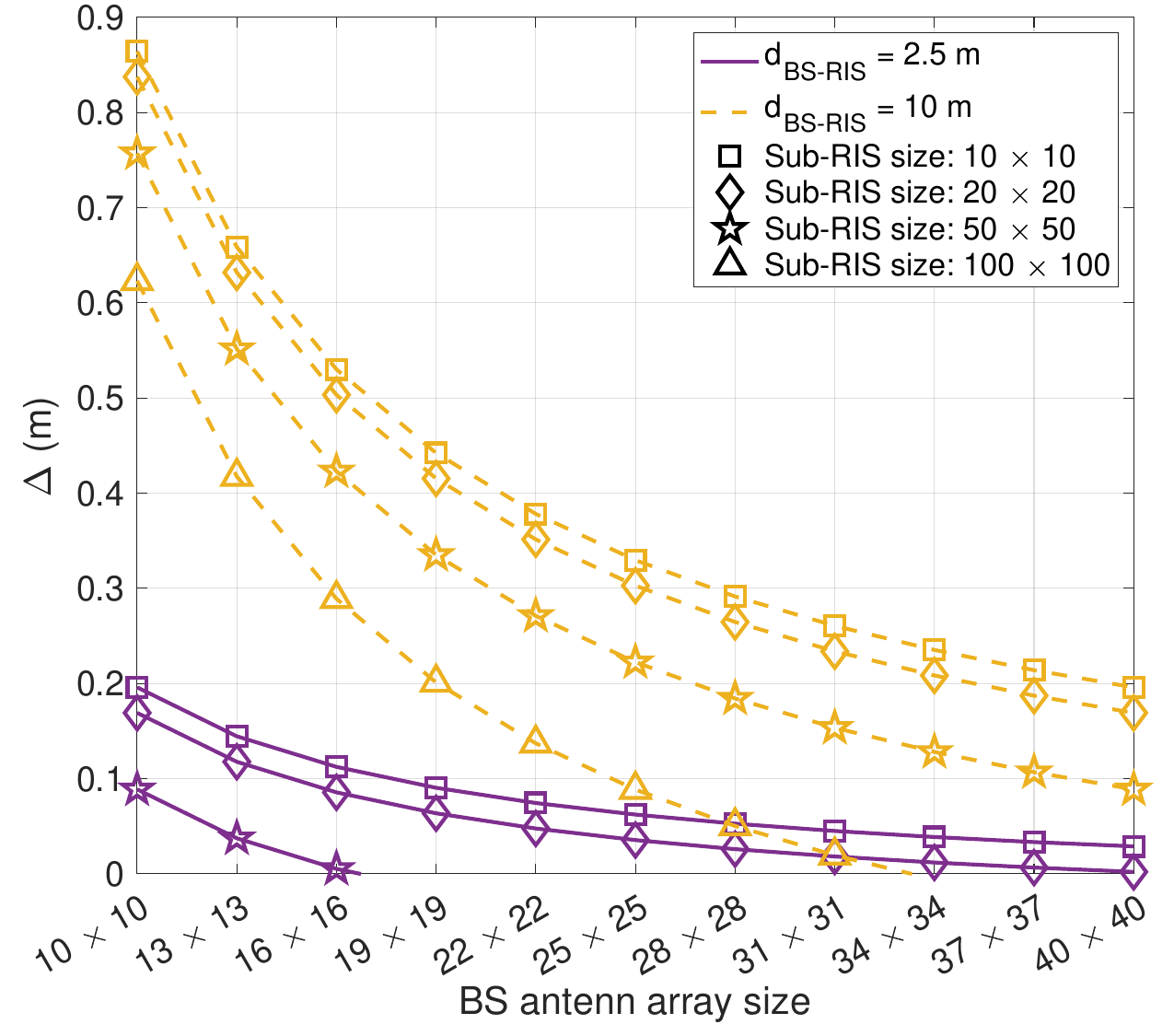}
    \caption{The variations of sub-RIS inter-spacing ($\Delta$) with respect to different BS antenna array sizes for different distances and sub-RIS sizes.}
    \label{fig: sub-RIS spacing}
    \vspace{-2em}
\end{figure}

In the simulation, we consider a limited and {known dead zone} \cite{etsigrris001}, as depicted in Figs. \ref{fig:scenarios}(a) and (b), which leads to an assumption that the elevation AoD of the sub-RISs is constrained within a limited range. Consequently, we consider that the azimuth and elevation AoDs of the sub-RISs are $\varphi_{t,ris} \in \mathcal{U}[-\pi, \pi]$ and $\vartheta_{t,ris} \in \mathcal{U}[-\pi/16, \pi/16]$, respectively. 
With such assumptions and by adopting a sub-RIS of size $10 \times 10$, the sub-RIS can cover a circular area with a diameter of $4 \ m$ which is located $10$ m away from the sub-RIS. 
Furthermore, given the fixed position of BS and sub-RISs, we can align the sub-RISs optimally with the BS \cite{10176315} to meet the desired minimum sub-RISs separation requirements. 
Based on our discussion in Section \ref{Sec: Plug-in RIS Structure} and as given in Fig. \ref{fig:EFD_IncidentAngle}(a), we assume that $\varphi_t, \varphi_{r,ris} \in \mathcal{U}[-\pi, \pi]$ and $\vartheta_t, \vartheta_{r,ris} \in \mathcal{U}[-\pi/3, \pi/3]$, considering the BS is equipped with a $10 \times 10$ antenna array. In other words, we assume that if the incident angle $\theta_0$ is confined to an interval of $[-\pi/3,\pi/3]$, the EFD remains constant\footnote{Note that different antenna arrays' and RISs' sizes, and the use of beam widening techniques, as discussed in \cite{8186193}, can lead to the definition of different AoD and AoA intervals at the terminals. Nevertheless, these aspects are beyond the scope of this paper and are left for future research directions.}. We can consider the maximum amount of EFD in this interval to calculate $\Delta$.

Fig. \ref{fig: sub-RIS spacing} depicts sub-RIS inter-spacing versus different system parameters. As shown in Fig. \ref{fig: sub-RIS spacing}, with increasing BS antenna array size, the minimum spacing among sub-RISs $\Delta$ decreases due to narrower beams emitted through BS. 
On the other hand, with increasing the distance between BS and RIS, the EFD increases, which means that the sub-RISs inter-space should be increased to avoid power leakage to the non-targeted sub-RISs.
Nevertheless, since the RIS is a passive device, we can increase the number of reflecting elements in each sub-RIS to decrease $\Delta$. This is specifically feasible in the proposed plug-in RIS since increasing the number of elements does not entail more complexity in the system. 
It is worth emphasizing that since the sub-RISs in the plug-in RIS exploit fixed beams, only traditional end-to-end channel estimation is enough, and cascaded channel decoupling is not required in the practical plug-in RIS system, resulting in significant complexity reduction in comparison with semi-passive and fully passive RIS structures, which allows us to easily increase the number of reflectors in the plug-in RIS structure. 
Fig. \ref{fig: sub-RIS spacing} shows that with increasing the sub-RIS size, the sub-RISs inter-space decreases. Note that in this paper, we consider BS antenna array and sub-RISs of size $10 \times 10$ to prevent high computational burden in the computer simulations, since in the computer simulations, we need to construct the cascaded channels separately for the sake of analysis.

\begin{figure}
    \centering
    \includegraphics[scale = 0.21]{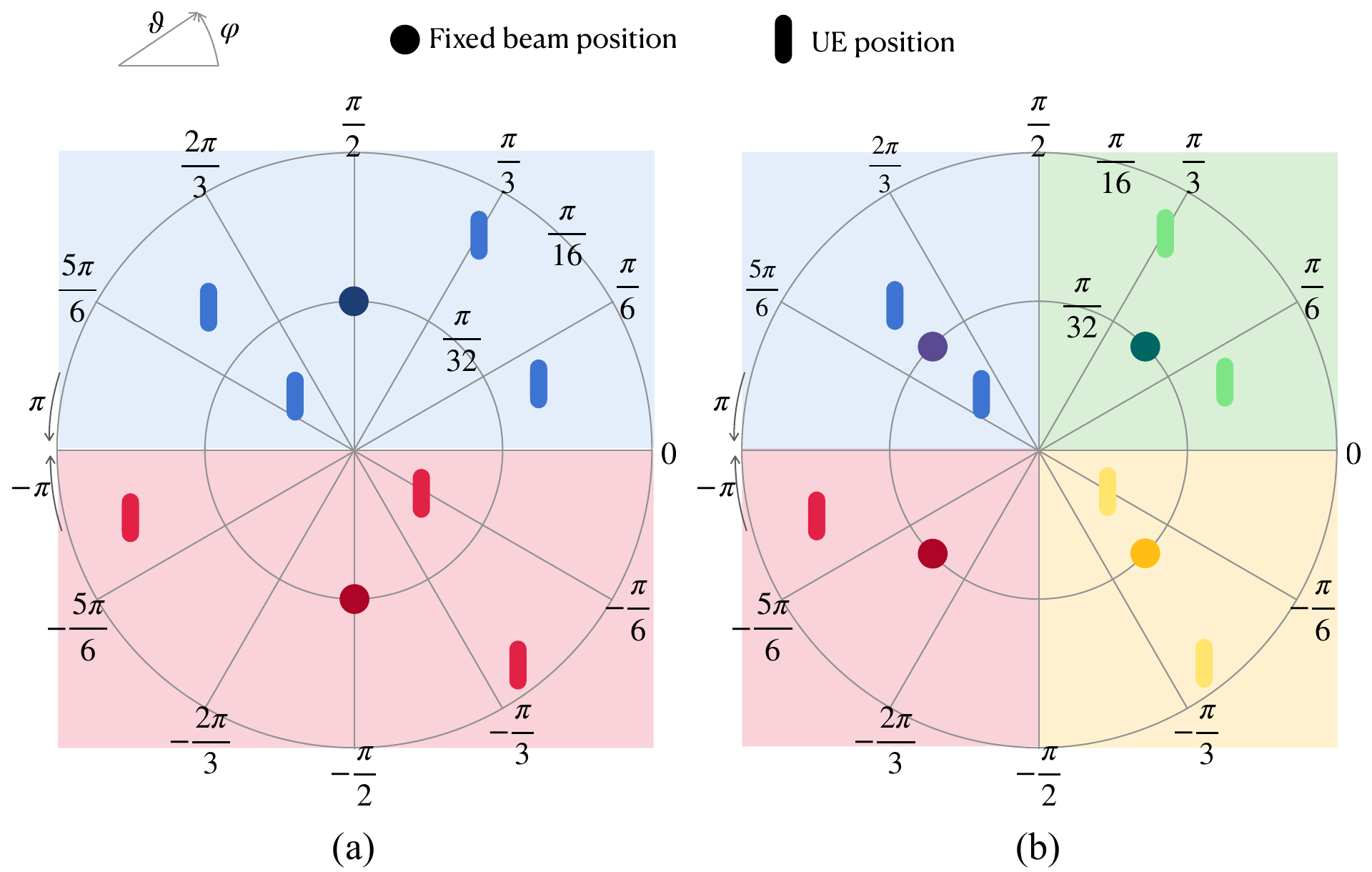}
    \caption{Dead zone divisions, fixed beams, and possible UE positions for (a) two segments and (b) four segments.}
    \label{fig:Zone Separation}
    \vspace{-1em}
\end{figure}

\begin{figure*}
    \centering
    \includegraphics[scale = 0.3]{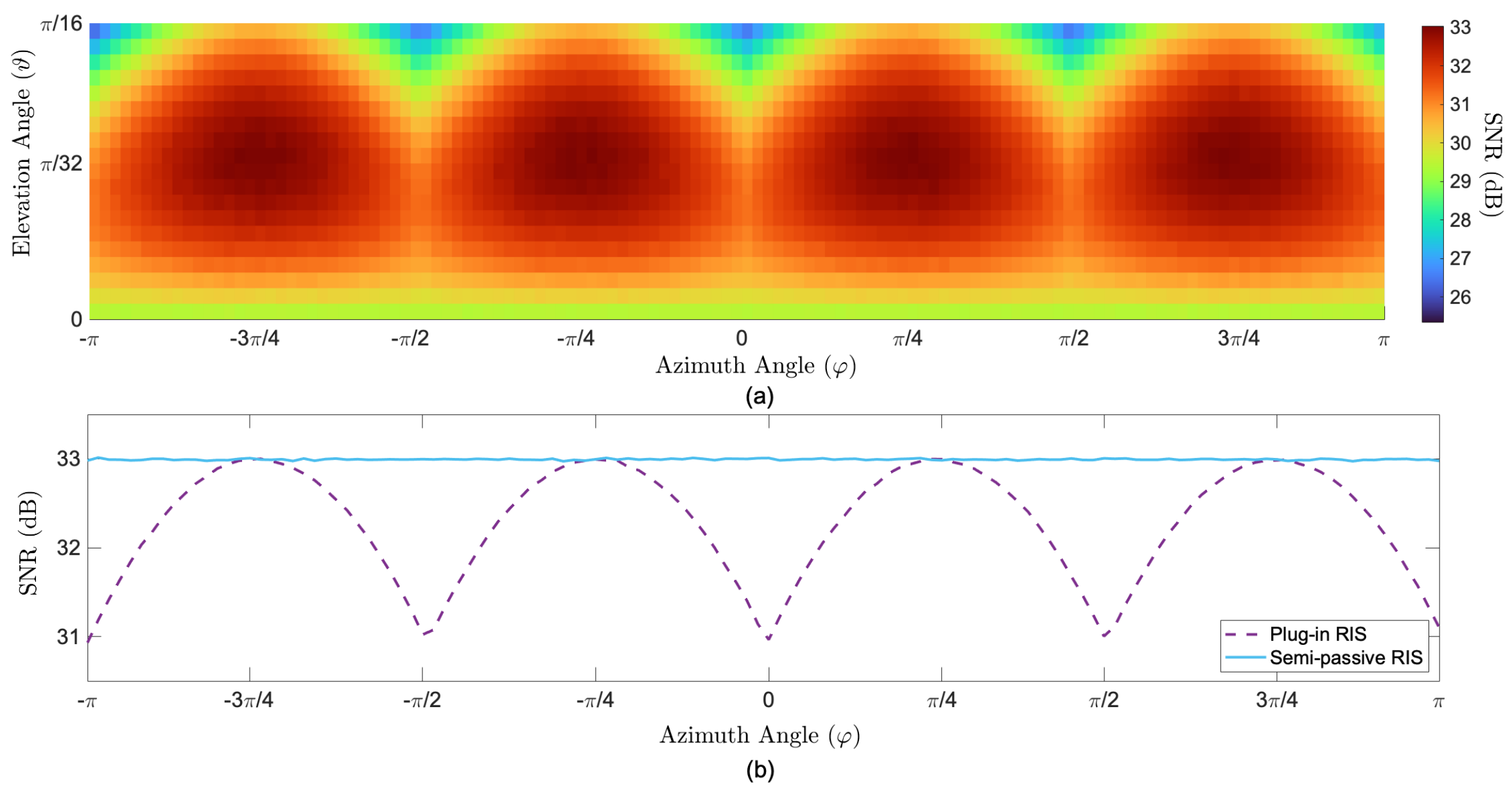}
    \caption{Coverage performance of the plug-in RIS; (a) considering all angular locations of the UE, (b) compared to the semi-passive RIS.}
    \label{fig:Coverage_Performance_R2}
    \vspace{-1.5em}
\end{figure*}

As depicted in Fig. \ref{fig:Zone Separation}, various segments within the dead zone are represented in polar coordinates for both two and four divisions. The fixed beam orientation for each segment and several potential UE positions within the dead zone are also illustrated in Fig. \ref{fig:Zone Separation}. It is worth mentioning that the information given in Fig. \ref{fig:Zone Separation} is used for plug-in RIS configuration and dead zone segmentation in this paper. Specifically, for two and four spatial segments, $(\phi_{b,i},\theta_{b,i}) \in \{ (\pi/2, \pi/32), (-\pi/2, \pi/32) \}$ and $(\phi_{b,i},\theta_{b,i}) \in \{ (\pi/4, \pi/32), (3\pi/4, \pi/32), (-\pi/4, \pi/32), (-3\pi/4, \pi/32) \}$, respectively.
As the number of segments increases, the average distance between the UE positions and the corresponding fixed beam orientations decreases; consequently, the UEs can receive signals via more favorable beams, which enhances overall system performance. It is important to mention that although Fig. \ref{fig:Zone Separation} shows dead zone division only along the azimuth angle, such division can also be applied along the elevation angle. However, for clarity and simplicity, we have opted not to illustrate more complex divisions in the figures.

It is noteworthy that in this paper, the system performance analysis is carried out using two and four sub-RIS configurations. For single-user scenarios, only one sub-RIS is used during each transmission period while the rest remain idle. On the other hand, the semi-passive RIS utilizes all implemented elements for maximum performance. In other words, the plug-in RIS uses only a portion of deployed elements, potentially sacrificing passive beamforming gain compared to the semi-passive RIS. However, the semi-passive RIS employs baseband components to estimate the cascaded channels and adjust its phase shifts autonomously, while the proposed plug-in RIS remains fully passive. Therefore, the plug-in RIS offers a favorable trade-off between system performance and EE/cost/complexity.
Likewise, for the blind RIS, all implemented elements are engaged in each transmission cycle.

\vspace{-1em}
\subsection{{Coverage Performance}}
\vspace{-0.5em}

{
    This subsection evaluates the plug-in RIS's coverage performance and compares it with the coverage performance of semi-passive RIS in the street canyon scenario. To do this, we calculate the received SNR for all UE's angular locations throughout the dead zone via Monte-Carlo simulation, as illustrated in Fig. \ref{fig:Coverage_Performance_R2}(a). The received SNR can be calculated as follows:
    \vspace{-0.5em}
    \begin{equation}
        \textrm{SNR} = \frac{|\sqrt{P}G_t G_r \mathbf{f}_r^H \mathbf{H}_{\text{eff},k} \mathbf{f}_t|^2}{\sigma^2}.
        \vspace{-0.5em}
    \end{equation}
    It is worth mentioning that, for this simulation, we considered four sub-RISs, each of size $10 \times 10$, along with $P = 10$ dBm. According to Fig. \ref{fig:Coverage_Performance_R2}(a), when UE is located at the center of each segment, the received SNR has its maximum amount of $33$ dB. The received SNR decreases as UE moves to the edges of each segment. By counting all the pixels given in Fig. \ref{fig:Coverage_Performance_R2}(a), it has been revealed that for around $85 \%$ of the UE's location spots, the received SNR is higher than $30$ dB (i.e., within $3$ dB of the SNR peak). Furthermore, as illustrated in Fig. \ref{fig:Coverage_Performance_R2}(a), the worst location for UE is the corner of each segment, as expected.}

{
    Fig. \ref{fig:Coverage_Performance_R2}(b) demonstrates the coverage performance of plug-in RIS compared to the benchmark, semi-passive RIS. For this simulation, we assumed $\vartheta = \frac{\pi}{32}$ for all the azimuth angles, and the semi-passive RIS is supposed to have a size of $10 \times 10$; hence, it has the same beamforming gain with the plug-in RIS in each transmission period.
    Due to the online configuration, the semi-passive RIS can provide a constant SNR of $33$ dB at the UE for all angular locations. The plug-in RIS performance is the same as semi-passive RIS only at the center of each spatial segment. By moving UE to the segment edges, the coverage performance degrades approximately $2$ dB and achieves $31$ dB at the cell edges.}

\vspace{-1.5em}
\subsection{ABER Performance}
\vspace{-0.5em}

This subsection displays simulation results for the ABER performance of the proposed plug-in RIS. We further verify these results via the analytical upper bound derived in Lemma \ref{Lem: CPEP}. Here, we utilize a $10 \times 10$ RIS for the semi-passive RIS structure to match the passive beamforming gain with the proposed plug-in RIS in each transmission period (since in each transmission period, only one sub-RIS of size $10 \times 10$ is used). In contrast, for the blind RIS design, we consider RIS sizes of $20 \times 10$ and $20 \times 20$ to highlight the superiority of the plug-in RIS over blind RIS. 

\begin{figure}
    \centering
    \includegraphics[scale = 0.27]{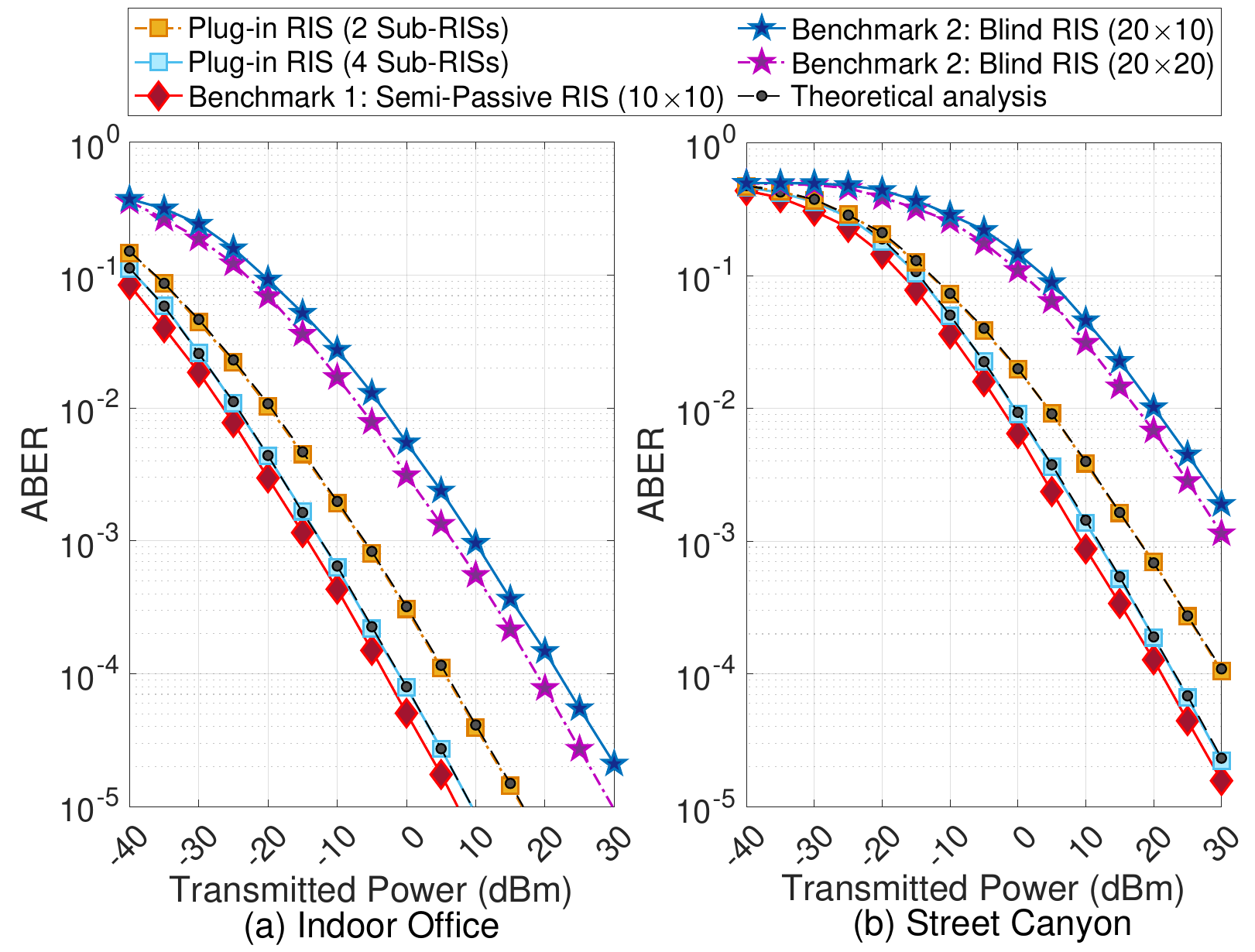}
    \caption{ABER performance analysis for (a) the BS-side RIS-aided system in the indoor office environment, (b) the UE-side RIS-aided system in the street canyon environment.}
    \label{fig:BER Sim}
    \vspace{-1em}
\end{figure}

As depicted in Fig. \ref{fig:BER Sim}, our simulations closely align with the upper bound, confirming the accuracy of the conducted simulations for both indoor office and street canyon scenarios. Notably, we utilized a binary phase shift keying (BPSK) signaling scheme for this simulation.
Fig. \ref{fig:BER Sim} illustrates a substantial performance improvement exhibited by the proposed plug-in RIS when compared with the blind RIS design, even though the blind RIS configuration considered here offers a larger beamforming gain due to its bigger size.
Additionally, results in Fig. \ref{fig:BER Sim} reveal that adopting a plug-in RIS with two sub-RISs results in roughly $9$ dB higher ABER compared to the semi-passive RIS. On the other hand, employing four sub-RISs leads to performance enhancement, causing only a $2$ dB ABER degradation compared to the semi-passive RIS in both scenarios.
It is important to recall that each sub-RIS corresponds to a segment within the dead zone, as shown in Fig. \ref{fig:Zone Separation}. Therefore, increasing the number of sub-RISs is equivalent to increasing the number of segments. Consequently, when we increase the sub-RISs from two to four, the ABER performance improves by about $7$ dB because the UE has more chance to receive a stronger signal.
These outcomes underline the efficacy of the proposed plug-in RIS structure, which proves to be more cost-effective than the semi-passive alternative with a few degradations in ABER.

\begin{figure}
    \centering
    \includegraphics[scale = 0.3]{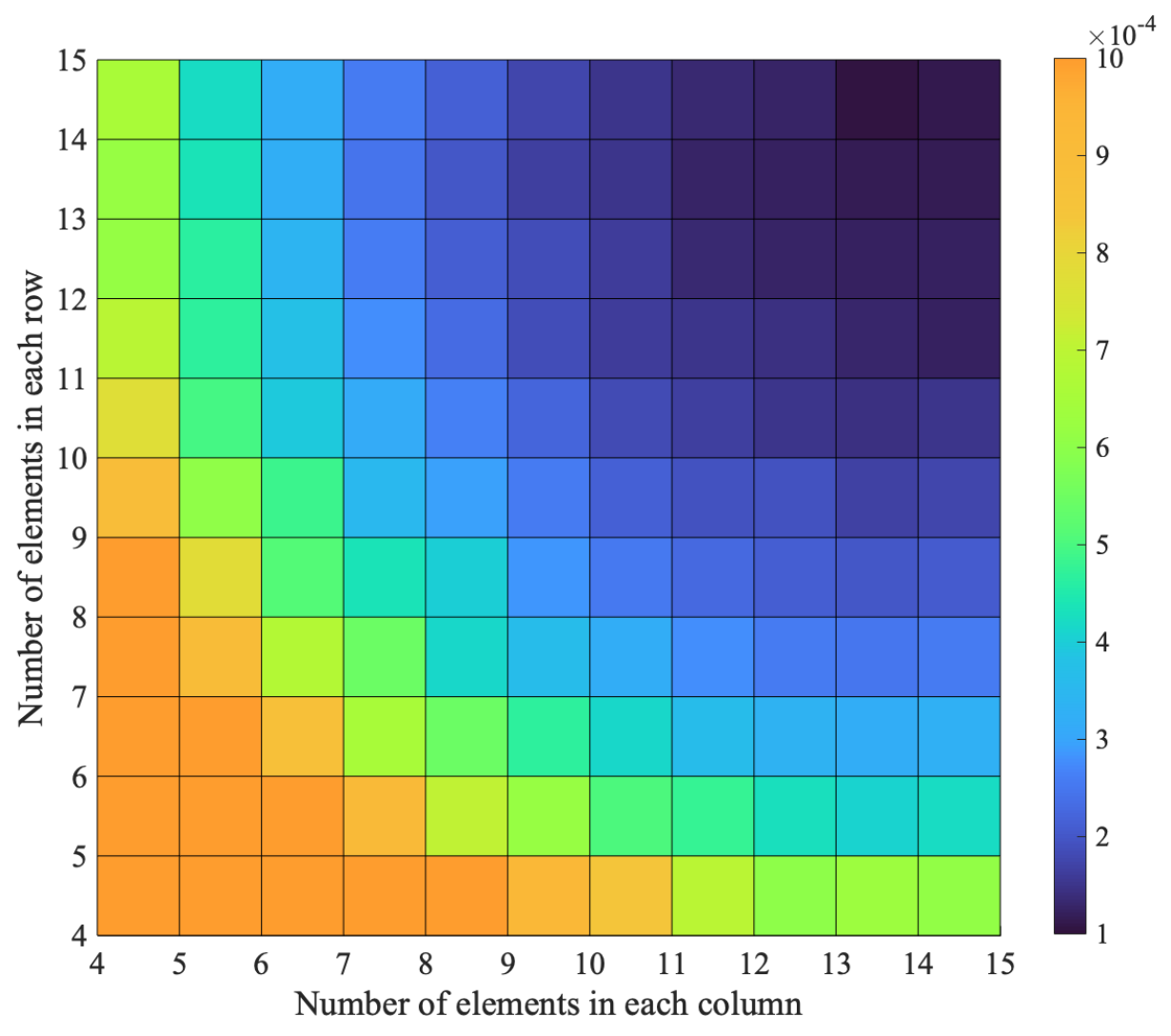}
    \caption{{ABER performance with increasing number of passive elements in each sub-RIS.}}
    \label{fig:Plug-in RIS vs Size}
    \vspace{-2em}
\end{figure}

{Fig. \ref{fig:Plug-in RIS vs Size} illustrates the ABER performance with the increasing number of passive elements in each sub-RIS in the plug-in RIS. For this simulation, we assumed that four sub-RISs have been deployed in the street canyon scenario, and the transmit power is considered to be $P = 20$ dBm. With the increasing number of passive elements, passive beamforming gain enhances, resulting in decreasing ABER.}

\vspace{-1.5em}
\subsection{Achievable Rate Performance}
\vspace{-0.5em}
In this subsection, we examine the achievable rate performance of the proposed plug-in RIS and compare it with benchmarks. For this subsection, alongside the RIS sizes explored in the previous subsection, we consider an RIS of size of $10 \times 10$ for the blind RIS scheme. 
The achievable rate can be calculated as 
\vspace{-0.5em}
\begin{equation}
    R = \mathbb{E} [ \log_2(1 + \textrm{SNR})].
    \vspace{-0.5em}
\end{equation}

Fig. \ref{fig:Achievable Rate Sim} depicts the plug-in RIS performance in terms of achievable rate compared to the two considered benchmark schemes. Similar to the ABER performance, the plug-in RIS exhibits satisfactory achievable rate performance compared to the {three} benchmark schemes under consideration. Compared with the blind RIS configuration, our novel plug-in RIS configuration significantly improves the achievable rate performance. Likewise, similar to the ABER performance outcomes, even larger-scale blind RIS configurations, which provide more significant beamforming gains, fail to reach the achievable rate performance offered by the plug-in RIS configuration.
In comparison to the semi-passive RIS, the proposed plug-in RIS shows a performance degradation of $5$ dB with $2$ sub-RISs. This degradation decreases to $2$ dB when using $4$ sub-RISs, demonstrating the impact of higher SNR at the UE due to increased segments in the dead zone. {Ultimately, AF relay performance is also given in Fig. \ref{fig:Achievable Rate Sim} as another benchmark. In this regard, we consider an AF relay equipped with a single antenna replaced with RIS in both indoor and outdoor scenarios. As illustrated in Fig. \ref{fig:Achievable Rate Sim}, for a fixed relay power $P_r$, with increasing BS transmitted power $P$, the achievable rate increases; however, after a certain level of $P$, the achievable rate remains constant. To explain this behavior, we refer to the relay's SNR, which can be calculated as follows \cite{8520809}:
\begin{equation}\label{eq: realy snr}
    \textrm{SNR}_{\textrm{relay}} = \frac{P_{r} G_t^2 P ||\mathbf{G} ||^2  |\mathbf{R}|^2}{ P_{r}|\mathbf{R}|^2 \sigma^2 + G_t^2 P ||\mathbf{G}||^2\sigma^2  + \sigma^4}.
\end{equation}
By increasing transmitted power $P$, term $G_t^2 P || \mathbf{G} ||^2 \sigma^2$ prevails over $P_{r} |\mathbf{R}|^2 \sigma^2 + \sigma^4$. Therefore, in high transmitted power $P$, we can write $\textrm{SNR}_{\textrm{relay}} \approx \frac{P_{r} G_t^2 {P} ||\mathbf{G} ||^2  |\mathbf{R}|^2}{G_t^2 {P} || \mathbf{G} ||^2 \sigma^2} = \frac{P_{r} G_t^2 ||\mathbf{G} ||^2  |\mathbf{R}|^2}{G_t^2 || \mathbf{G} ||^2 \sigma^2}$; hence, BS transmitted power has no effect on the AF relay's SNR when surpasses a specified level. On the other hand, with increasing relay's power $P_r$, the achievable rate increases, as shown in Fig. \ref{fig:Achievable Rate Sim}. Nonetheless, due to the constant achievable rate performance of the AF relay at the high BS transmitted power, it cannot surpass plug-in RIS in such power levels.}

Computer simulation results in this subsection emphasize the effectiveness of our plug-in RIS compared to the benchmarks. While the plug-in RIS may not outperform the semi-passive RIS, its cost-effective passive design merits consideration.

\vspace{-1.5em}
\subsection{Energy Efficiency}
\vspace{-0.5em}

\begin{figure}
    \centering
    \includegraphics[scale = 0.25]{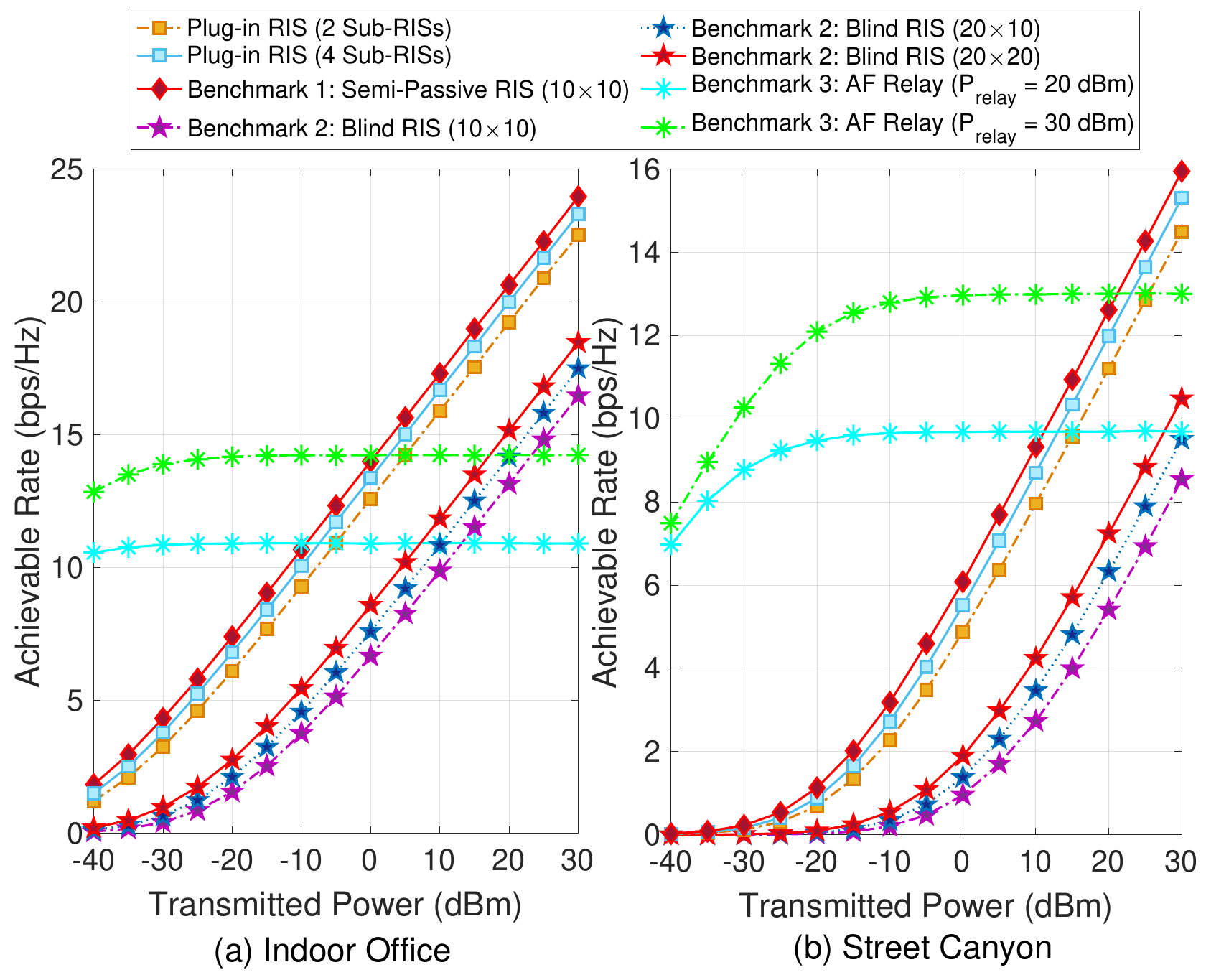}
    \caption{Achievable rate performance for (a) the BS-side RIS-aided system in the indoor office environment, (b) the UE-side RIS-aided system in the street canyon environment.}
    \label{fig:Achievable Rate Sim}
    \vspace{-1.5em}
\end{figure}

In this subsection, we focus on the strength point of the plug-in RIS design, which is its EE. We also highlight the trade-offs between EE and ABER/achievable rate. The EE is computed as follows:
\vspace{-0.5em}
\begin{equation}
    \eta_{EE} = \frac{R \times B}{P_c} \ \text{bits/Joule},
    \vspace{-0.5em}
\end{equation}
where $P_c$ signifies the power consumption within the system and can be determined as follows \cite{9370097}:
\vspace{-0.5em}
\begin{equation}
    P_c = P_{\text{Tx}} + P_{\text{Rx}} + P_{\text{RIS}},
    \vspace{-0.5em}
\end{equation}
where $P_i$ ($i \in \{ \text{Tx, Rx, RIS} \}$) denotes the power consumption of the respective terminal and can be calculated as follows \cite{7370753, 7876856}:
\vspace{-0.5em}
\begin{equation}
\begin{split}
    P_{\text{Tx}} = P + N_t & P_{\text{PA}} + N_{\text{rf}} ( N_t P_{\text{PS}} 
    + P_{\text{RF-chain}} + 2P_{\text{DAC}} ) + P_{\text{BB}},
\end{split}
\end{equation}
\begin{equation}
\begin{split}
    P_{\text{Rx}} = N_r P_{\text{LNA}} & + N_{\text{rf}} (N_r P_{\text{PS}} 
    + P_{\text{RF-chain}} + 2P_{\text{ADC}}) + P_{\text{BB}},
\end{split}
\end{equation}
\begin{equation}
    P_{\text{Plug-in RIS}} = P_{\text{Blind RIS}} = M P_{\text{PA\_RIS}},
\end{equation}
\begin{equation}
    {P_{\text{Codebook-based RIS}} = P_{\text{controller}} + M P_{\text{PA\_RIS}},}
\end{equation}
\begin{equation}
\begin{split}
        & P_{\text{Semi\_Passive RIS}} =  P_{\text{controller}} + M P_{\text{PA\_RIS}} \\
        & + M_{\text{active}}(P_{\text{LNA}} + P_{\text{RF-chain}} + 2 P_{\text{ADC(RIS)}}) + P_{\text{BB}},
\end{split}
\vspace{-0.5em}
\end{equation}
where $M$ represents for number of RIS/sub-RIS elements engaged in the communication and $P_{\text{ADC}}$ ($P_{\text{DAC}}$) is the consumption power of analog-to-digital (ADC) (digital-to-analog (DAC)) converter at the receive (transmit) terminal and can be computed as follows \cite{9370097, 7876856}:
\begin{equation}
    P_{\text{ADC}} = P_{\text{DAC}} = FOM_W \times f_s \times 2^b,
    \vspace{-0.5em}
\end{equation}
where $FOM_W$ corresponds to Walden's figure-of-merit for assessing ADC power efficiency, the variable $f_s$ stands for the Nyquist sampling frequency, while $b$ represents the ADC resolution bits. Our assumption employs $FOM_W = 46.1$ fJ/conversion-step for a $100$ MHz bandwidth \cite{9370097}. We utilize ADCs with $1$ bit and $4$ bits resolution for the semi-passive RIS \cite{9685434} and transceivers \cite{7876856}, respectively. We also consider $8\%$ of elements in the semi-passive RIS being connected to baseband components as suggested in \cite{9685434}. All other parameters align with those defined in Table \ref{tab:System parameters}.

\begin{figure}
    \centering
    \includegraphics[scale = 0.3]{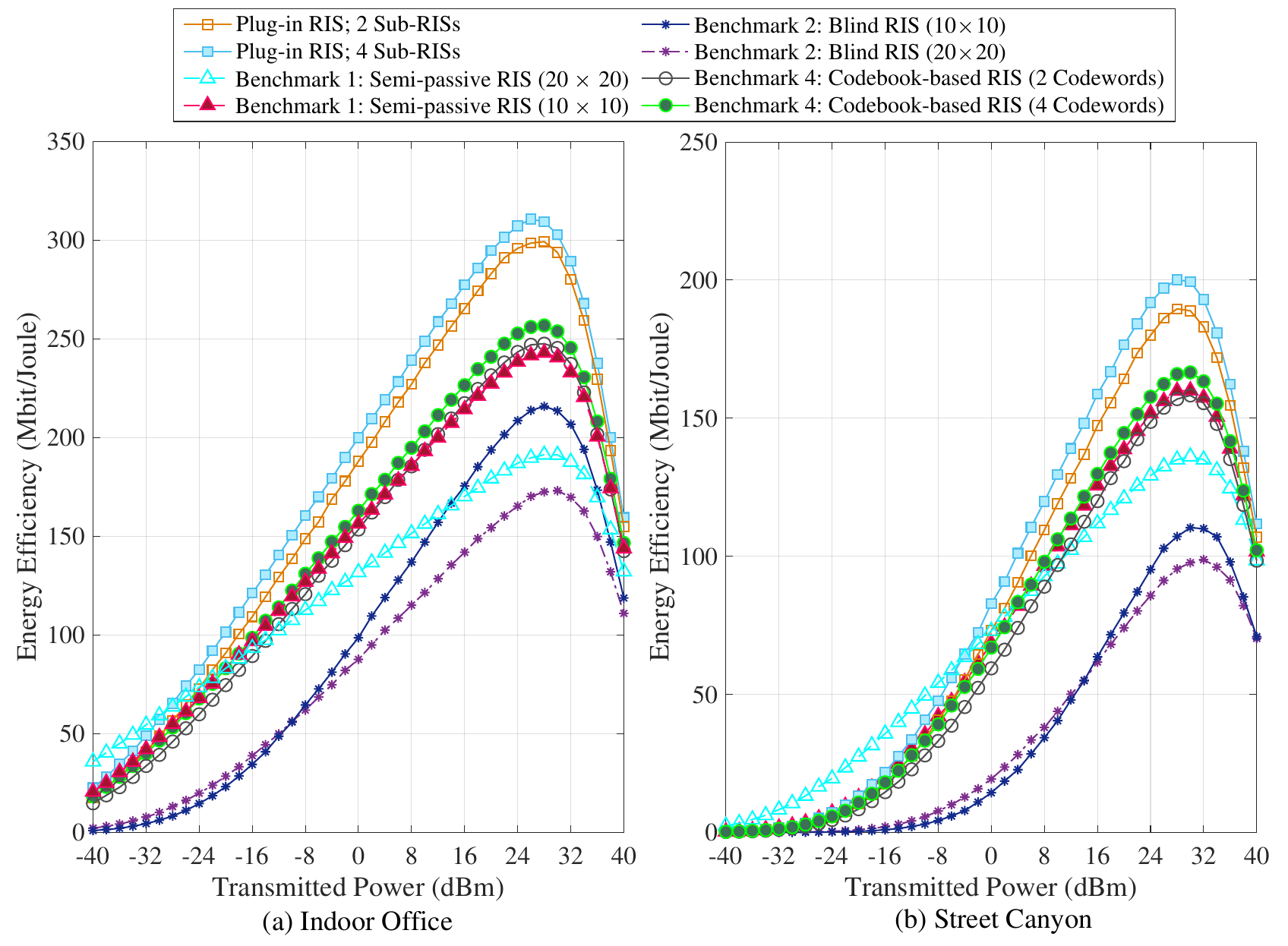}
    \caption{Energy efficiency comparison for (a) the BS-side RIS-aided system in the indoor office environment, (b) the UE-side RIS-aided system in the street canyon environment.}
    \label{fig:EE Sim}
    \vspace{-1em}
\end{figure}

As illustrated in Fig. \ref{fig:EE Sim}, the plug-in RIS performs better than all benchmarks in indoor and outdoor scenarios. The performance of the plug-in RIS improves as the number of sub-RISs increases, benefiting from higher SNR and the power-efficient nature of passive elements. Essentially, increasing the number of sub-RISs allows for increasing the number of segments within the dead zone. Consequently, each segment becomes smaller, leading to improved received SNR and better $\eta_{EE}$.
On the contrary, increasing the number of elements in the semi-passive RIS negatively affects EE and leads to performance deterioration. It is important to note that although enlarging the RIS size in the semi-passive RIS configuration increases received SNR due to enhanced beamforming gain, it also involves more baseband components in the RIS structure, resulting in higher power consumption. In this case, power consumption primarily impacts EE performance and leads to decreased EE.

Similarly, in the case of the blind RIS, increasing the number of elements reduces EE due to the dominant effect of power consumption by passive elements. In other words, in the blind RIS, increasing the passive elements has a more negative effect on power consumption than its positive effect on beamforming gain.
It is worth mentioning that a notable EE performance gap exists between blind RIS and the proposed plug-in RIS, primarily due to the enhanced SNR provided by the plug-in RIS without a proportional increase in power consumption.
{The EE performance of the codebook-based RIS is also depicted in Fig. \ref{fig:EE Sim}, which is worse than the plug-in RIS due to the controller power consumption of codebook-based RIS. It is worth mentioning that the higher the number of codewords, the more the EE performance is enhanced, similar to the plug-in RIS, which indicates better performance with increasing the number of sub-RISs.}

\vspace{-1.5em}
\subsection{Detector Complexity}
\label{subsec: SIM RES-Detector complexity}
\vspace{-0.5em}

\begin{figure}
    \centering
    \includegraphics[scale = 0.29]{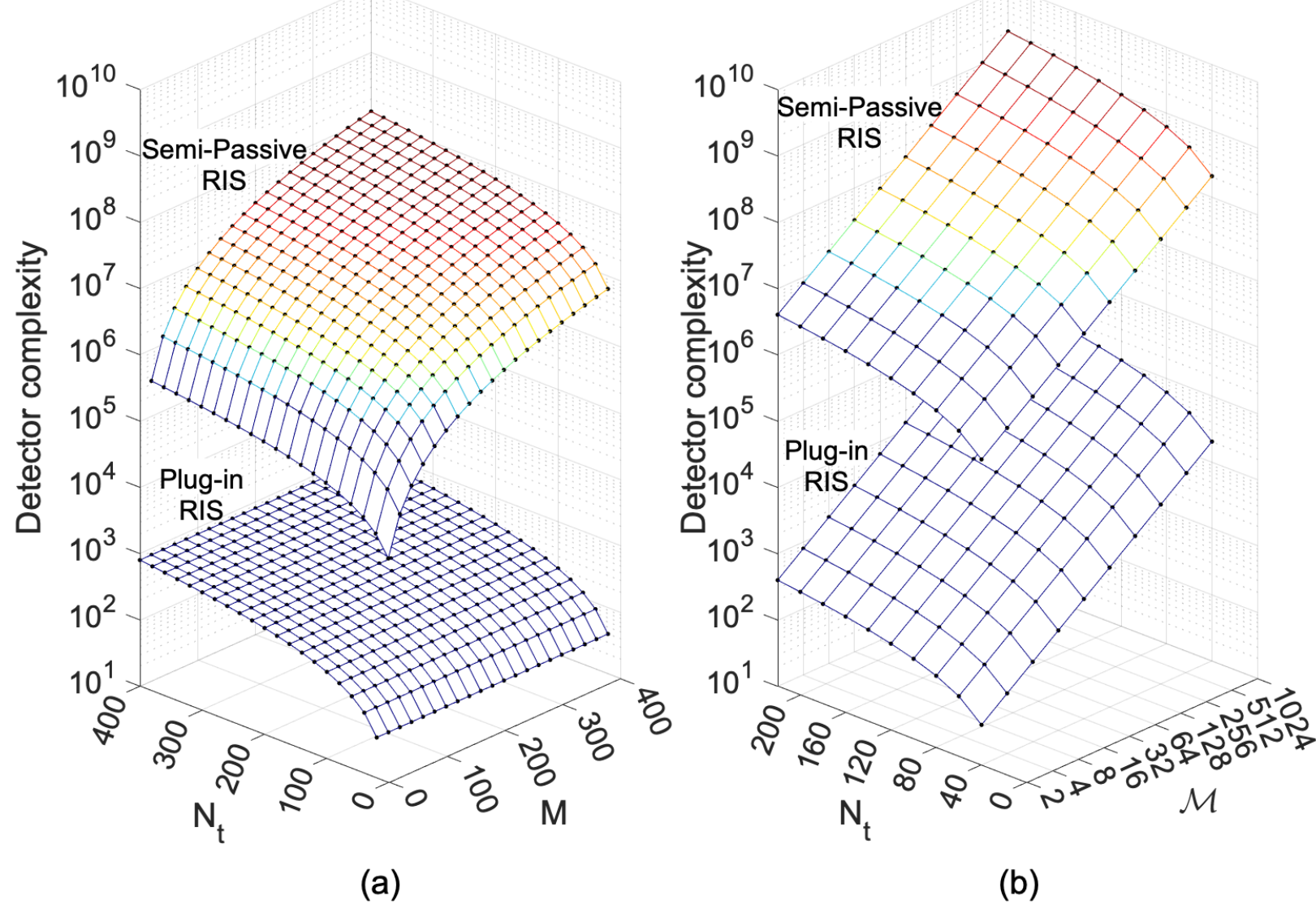}
    \caption{Detector complexity comparison as a function of (a) number of transmit antennas, $N_t$, and number of RIS elements, $M$, (b) number of transmit antennas, $N_t$, and constellation order, $\mathcal{M}$.}
    \label{fig: Detector Complexity}
    \vspace{-1.5em}
\end{figure}

The detector complexity for different numbers of transmit antennas, RIS elements, and constellation orders is depicted in Fig. \ref{fig: Detector Complexity}.
As shown in Fig. \ref{fig: Detector Complexity}(a), it is evident that the plug-in RIS's detector exhibits significantly lower computational complexity compared to the semi-passive RIS's detector. It is worth noting that the increase in the number of RIS elements ($M$) has a more significant impact on the complexity of the semi-passive RIS detector compared to the increase in the number of antennas ($N_t$); this is attributed to the quadratic relationship between $M$ and the detector complexity in the semi-passive detector expression, while the plug-in RIS design remains unaffected by the increase in $M$. Note that in Fig. \ref{fig: Detector Complexity}(a), we have considered $\mathcal{M} = 2$.
Fig. \ref{fig: Detector Complexity}(b) further illustrates how variations in $N_t$ and $\mathcal{M}$ affect detector complexity while maintaining a constant number of RIS elements ($M = 100$), emphasizing the superiority of the proposed plug-in RIS over semi-passive RIS design.

\vspace{-1em}
{
\section{Illustrative Results - Multi User Case}\label{Sec:Illustrative Results (MU)}
\vspace{-0.5em}
\begin{figure}
    \centering
    \includegraphics[scale = 0.3]{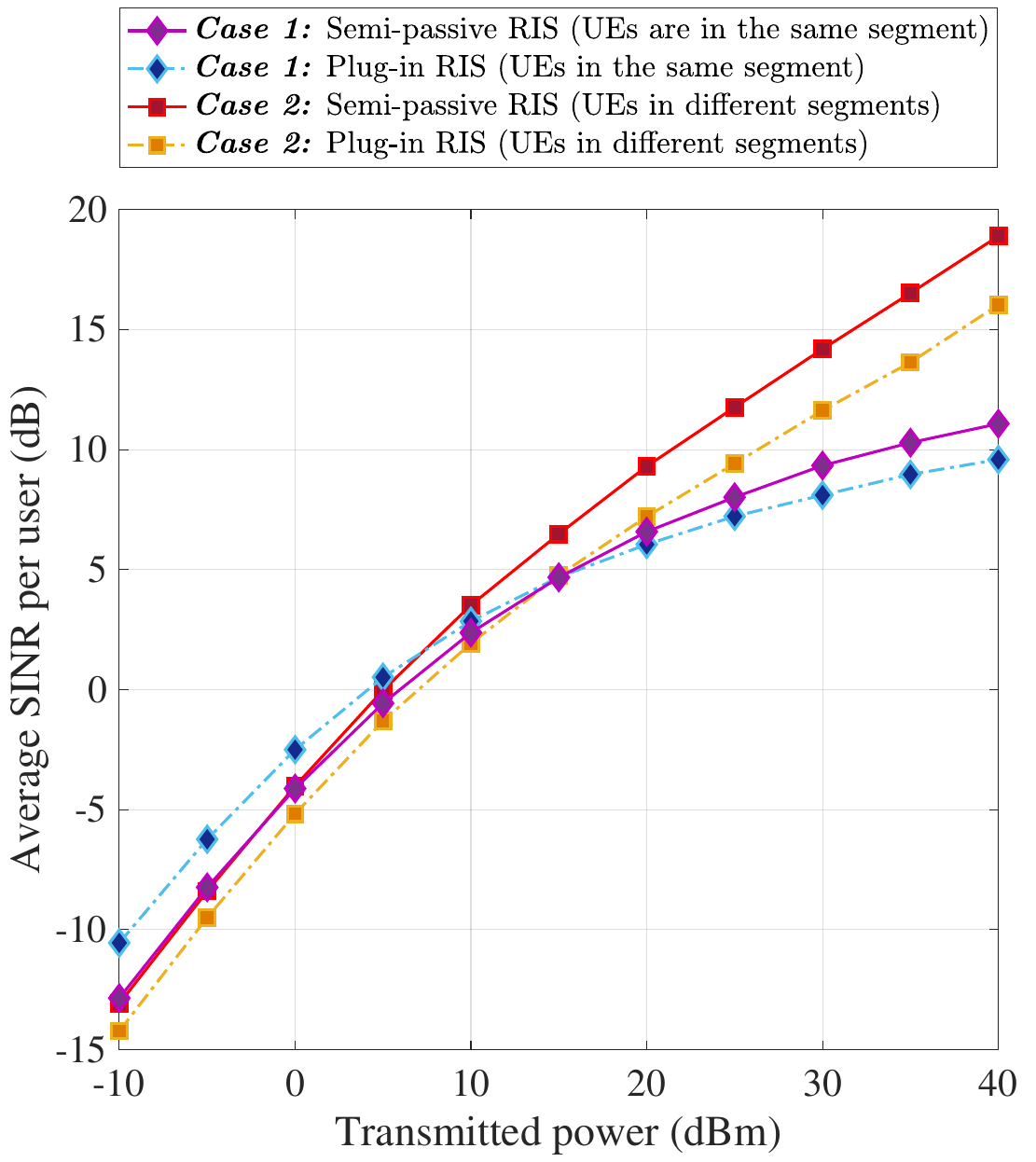}
    \caption{Average SINR performance of the plug-in RIS compared to the semi-passive RIS for a two-UE scenario.}
    \label{fig:MU_Avg_SINR}
    \vspace{-1em}
\end{figure}

This section investigates the SINR performance of the proposed plug-in RIS in an MU scenario. We consider a system setup with two UEs, two spatial segments, and the street canyon environment and investigate two different cases to ensure a more precise analysis: in case 1, both UEs are located in the same segment, while in case 2, they are situated in different segments. 
Note that, for a fair comparison, we consider the same RIS size for semi-passive RIS and each sub-RIS, i.e., $10 \times 10$; hence, in each transmission period, the same number of reflectors engaged in the communication for both schemes, resulting in the same passive beamforming.

Fig. \ref{fig:MU_Avg_SINR} compares the SINR    performance of MU-plug-in RIS and MU-semi-passive RIS.
It is important to highlight that, in the case of MU-semi-passive RIS, we employed one RIS for each UE to ensure effective passive beamforming. Sharing a single semi-passive RIS among two UEs results in the beamforming gain of the RIS being divided between them, thereby declining the effectiveness of passive beamforming.
As shown in Fig. \ref{fig:MU_Avg_SINR}, at high transmitted power levels,  plug-in RIS experiences $1.5$ dB and $3$ dB performance loss compared to the benchmark for case 1 and case 2, respectively. 
The lower performance loss of case 1 can be attributed to the better performance of plug-in RIS in this scenario. The plug-in RIS in case 1 only exploits one sub-RIS for serving both UEs; hence, there is only interference in the BB stage, while in the other scenarios, two RISs/sub-RISs are adopted, which entails interference in the analog stage as well. Accordingly, the RZF BB precoder works more efficiently in plug-in RIS under case 1's setup. In order to cancel interference in both analog and BB stages, more sophisticated precoder designs like joint-group-processing (JGP) and common-group-processing (CGP) are required, as described in \cite{8891298,9086460}. The interference effect at the analog stage is more dominant at the low power levels; hence, plug-in RIS under case 1's setup performs better than the other setups.
}

\vspace{-1em}
\section{Conclusion}
\label{Sec: Conclusion}
\vspace{-0.5em}

This paper has introduced a practical RIS structure, the plug-in RIS, for mmWave communication systems to enhance coverage to the dead zones. The plug-in RIS operates passively, cleverly integrating the control mechanism within the transmitted beam to the RIS, eliminating the need for conventional reliable control links. It also relaxes the channel estimation process by eliminating complex cascaded channel decoupling, a common challenge in RIS-assisted systems.
In this approach, dead zones are divided into segments, each served by a dedicated sub-RIS with a fixed beam. Computer simulation results have shown that deploying four sub-RISs causes only slight degradation in ABER and achievable rates, making fully passive operation feasible.
We have also compared the EE of the plug-in RIS with benchmarks. While the semi-passive RIS slightly outperforms the plug-in RIS in terms of ABER and achievable rate, its EE performance is worse than the proposed plug-in RIS due to active baseband components. Besides, the plug-in RIS detector exhibits superior complexity performance compared to the semi-passive RIS thanks to the conventional channel estimation mechanism. {Ultimately, extending the plug-in RIS into an MU scenario has also been investigated by studying the average SINR performance compared to the semi-passive RIS. It has been revealed that adopting plug-in RIS in an MU scenario only results in a few dB of performance loss compared to the semi-passive RIS. Nevertheless, MU solutions to mitigate interference can be future research directions that require more sophisticated design scenarios by considering advanced algorithms.}

In summary, our plug-in RIS proves to be a compelling solution, performing closely to the semi-passive RIS regarding ABER and achievable rate performances, surpassing the benchmarks in terms of EE performance and detector complexity, and addressing challenges in mmWave systems. 

\vspace{-1em}

\appendices

\section{Proof of Corollary \ref{Cor: EFD Approximation}}
\label{Apndx: EFD Aprox Proof}
\vspace{-0.5em}
By substituting $\theta_0 = 0$ into (\ref{eq: EFD exact}), we can simplify it as:
\begin{equation}\label{eq: EFD tangent}
    \textrm{EFD} = 2 d \times \tan \left( \frac{\textrm{HPBW}}{2} \right).
\end{equation}
On the other hand, the tangent function can be expanded using the MacLaurin series as follows:
\begin{equation}
    \tan(x) = x + \frac{x^3}{3} + \frac{2x^5}{15} + \dots ; \ \ \textrm{if} \ \ \lvert x \rvert  < \frac{\pi}{2},
\end{equation}
whereas when $x$ is of a small magnitude, it can be approximated as $\tan(x) \approx x$.
Note that $\frac{\textrm{HPBW}}{2}$ becomes relatively small when implementing a large array at the BS. Consequently, the approximation formula remains applicable here, and we can further simplify (\ref{eq: EFD tangent}) to (\ref{eq: EFD approximation}).

\vspace{-1em}

\section{Proof of Lemma \ref{Lem: CPEP}}
\label{Apndx: CPEP proof}
By exploiting \cite[equation (24)]{10176315} and after some mathematical manipulations, (\ref{eq: 1st equation CPEP}) can be updated as 
\begin{equation} \label{eq: CPEP Apndx}
    \begin{split}
        \mathcal{P}( s^*  \rightarrow & \hat{s} | \alpha_k, \beta_k ) 
        = \mathcal{P} ( |\sqrt{P} G_t G_r \mathbf{f}_r^H \mathbf{H}_{\text{eff},k} \mathbf{f}_t (s^* - \hat{s})|^2 \\
        & + 2 \mathcal{R}\{ \sqrt{P} G_t G_r \mathbf{n}^H \mathbf{f}_r \mathbf{f}_r^H \mathbf{H}_{\text{eff},k} \mathbf{f}_t (s^* - \hat{s}) < 0 \}).
    \end{split}
    \vspace{-2em}
\end{equation}
As the elements of $\mathbf{n}^H$ follow a complex normal distribution with variance $\sigma^2$, its real component also conforms to a normal distribution with variance $\frac{\sigma^2}{2}$, represented as $\mathcal{R}\{\mathbf{n}^H\} \sim \mathcal{N}(\mathbf{0}_{N_r}, \frac{\sigma^2}{2}\mathbf{I}_{N_r})$. Accordingly, calculating (\ref{eq: CPEP}) is straightforward.





\ifCLASSOPTIONcaptionsoff
  \newpage
\fi

\vspace{-1em}

\bibliographystyle{ieeetr}
\bibliography{reference}




\end{document}